\documentclass[]{interact}

\usepackage{epstopdf}

\usepackage[numbers,sort&compress]{natbib}
\bibpunct[, ]{[}{]}{,}{n}{,}{,}
\renewcommand\bibfont{\fontsize{10}{12}\selectfont}

\usepackage{color}
\theoremstyle{plain}
\newtheorem{theorem}{Theorem}[section]

\usepackage{multirow}
\theoremstyle{definition}

\theoremstyle{remark}

\usepackage{hyperref}
\usepackage{csquotes}
\usepackage{lscape}
\usepackage[format=hang]{caption}

\begin{document}


\title{A prediction criterion for working correlation structure selection in GEE}

\author{
\name{Gul~Inan\textsuperscript{a}, \thanks{CONTACT Gul~Inan. Email: ginan@metu.edu.tr} Mahbub~A.H.M.~Latif\textsuperscript{b} and John~Preisser\textsuperscript{c}}
\affil{\textsuperscript{a} Middle East Technical University, Ankara, Turkey;
\textsuperscript{b}Institute of Statistical Research and Training, University of Dhaka, Dhaka, Bangladesh;
\textsuperscript{c}Department of Biostatistics, University of North Carolina, Chapel Hill, U.S.A.}
}

\maketitle

\begin{abstract}
Generalized estimating equations (GEE) is one of the most commonly used methods for marginal regression analysis of longitudinal data, especially with discrete outcomes. The GEE method models the association among the responses of a subject through a working correlation matrix and correct specification of the working correlation structure ensures efficient estimation of the regression parameters. This study proposes a predicted residual sum of squares (PRESS) statistic as a working correlation selection criterion in GEE. An extensive simulation study is designed to assess the performance of the proposed GEE PRESS criterion and to compare its performance with well-known existing criteria in the literature. The results show that the GEE PRESS criterion has better performance than the weighted error sum of squares SC criterion in all cases and is comparable with that of other existing criteria when the true working correlation structure is AR(1) or exchangeable. Lastly, the working correlation selection criteria are illustrated with the Coronary Artery Risk Development in Young Adults study.
\end{abstract}


\section{Introduction}
\label{sec:intro}

Longitudinal studies arise from repeated measurements on a given response for the same subjects over time. Longitudinal studies are conducted in different fields such as clinical trials, epidemiology, behavioural sciences, econometrics and so on. Generalized estimating equations (GEE) \citep{liang1986longitudinal} is one of the most commonly used methods for marginal analysis of longitudinal data, especially with discrete responses. In particular, GEE, which does not require specification of a joint distribution for the repeated measurements of a subject, estimates a generalized linear model for the marginal mean of each measurement of the subject. The method of GEE further models the correlation within the repeated measurements of the response of a subject through a working correlation matrix, such as first-order autoregressive (AR(1)), exchangeable (Exch), and unstructured (UN) correlation matrices.

Irrespective of whether or not the working correlation structure type is correctly specified, the GEE approach produces consistent estimates for regression parameters, with a correctly specified marginal mean model \citep{liang1986longitudinal}. However, correct specification of the working correlation structure ensures efficient estimation of the regression parameters within the class of linear unbiased estimating equations. For this reason, developing methods for working correlation structure selection in GEE analysis, conditional on the correctly specified marginal mean model, has been an active area of research and, in turn, several criteria for working correlation structure selection in GEE analysis have been proposed.
 
\citet{pan2001akaike} proposed the quasi-likelihood information criterion (QIC), an extension of the most commonly used model selection criterion for classical linear regression, the Akaike's information criterion. \citet{hin2009working} modified QIC to propose the correlation information criterion ({CIC}) and found that it improves correlation structure selection relative to the QIC.  QIC and  CIC contrast the empirical covariance matrix for the regression parameters to the model-based covariance matrix that assumes independence. \citet{rotnitzky1990hypothesis} introduced the criterion labeled {RJ} and \citet{gosho2011criterion} found that neither {RJ} criterion or {QIC} is superior to the other. \citet{shults1998analysis} proposed a criterion named {SC} for working correlation structure selection that minimizes the weighted residual sum of squares of the response vector. \citet{shults2009comparison} found that {RJ} criterion performed better than {SC} criterion in selecting the correlation structure in analyzing correlated binary responses.  \citet{gosho2011criterion} proposed a criterion that is a function that contrasts the unweighted sum of squares, or empirical covariance matrix of the response vector, with the model-based covariance of the responses. In their simulation study, their criterion was compared with CIC, QIC, RJ, and SC criteria and none was found best for all cases considered.  

The aim of this study is twofold. First, we propose a predicted residual sum of squares (PRESS) statistic for selection of working correlation structure in GEE analysis, where the particular form of the proposed PRESS criterion is similar to the SC criterion, with a different weight matrix. Afterwards, we design an extensive simulation study and compare the performance of the proposed PRESS criterion with several existing criteria such as {CIC}, {QIC}, {RJ}, and {SC} for balanced and unbalanced longitudinal data with discrete responses (e.g., binary and Poisson outcomes) under several working correlation structures such as Exch, AR(1), and UN.

The subsequent sections of the paper are organized as follows. In Section~\ref{sec:gee}, GEE analysis and existing working correlation structure selection criteria are briefly reviewed. In Section~\ref{sec:press}, a PRESS statistic for GEE is introduced and its connections to prediction error and cluster-deletion diagnostics are presented. Section~\ref{sec:sim} contains the simulation study evaluating the performance of the PRESS statistic relative to other criteria in selecting the working correlation structure for balanced and unbalanced longitudinal data with discrete responses in marginal regression models estimated with GEE. Section~\ref{sec:example} presents the application of PRESS statistic to analysis of smoking trends within the Coronary Artery Risk Development in Young Adults (CARDIA) study. Finally, Section~\ref{sec:conclusion} states limitations of the current work, and proposes future research.


\section{Overview of GEE and existing working correlation structure selection criteria in GEE}
\label{sec:gee}

In this section, we give a brief overview of the GEE method and several commonly used working correlation structure selection criteria in the literature. 

\subsection{Overview of GEE}

Suppose that there are {N} subjects ($i=1,\ldots,N$) in a longitudinal study and each subject is observed at times $t=1,\ldots,n_{i}$. For the $i$th subject, let $\mathbf{y}_i=(y_{i1}, \ldots, y_{in_i})'$ be the $n_i$-dimensional response vector and $\textbf{X}_i$ be the corresponding covariate matrix of order $n_i\times p$. Let the marginal mean response vector for the $i$th subject be given by $E(\mathbf{y}_i|\textbf{X}_i)=\bm{\mu}_i$ and $g(\bm{\mu}_i)={\textbf{X}}^{'}_i\bm{\beta}$, where $\bm{\beta}=(\beta_1, \ldots, \beta_p)'$ is vector of regression parameters and $g(\cdot)$ is the link function (e.g., \textit{logit} and $\log$ link functions are used for binary and Poisson responses, respectively). Assume that the $n_i \times n_i$-dimensional working covariance matrix of response vector for the $i$th subject is defined as:

\begin{equation}
\textbf{V}_i=\phi \textbf{A}_i^{1/2}\textbf{R}_i \textbf{A}_i^{1/2}, \nonumber
\end{equation}

\noindent
where $\phi$ is a positive-valued scale parameter and is equal to $1$ when there is no over-dispersion in the response,  $\textbf{A}_i=\text{diag}\{h(\mu_{i1}),\ldots,h(\mu_{in_i})\}$ is a diagonal matrix, $h(\cdot)$ is the variance function (e.g., for binary responses $h(\mu_{it})=\mu_{it}(1-\mu_{it})$  and for Poisson responses $h(\mu_{it})=\mu_{it}$), and $\textbf{R}_i=\textbf{R}_i(\bm{\alpha})$ is the working correlation matrix representing the association between the repeated measurements of the outcome for a subject, which is completely specified by the finite-dimensional parameter vector $\bm{\alpha}$.

The estimate of the regression parameters $\bm{\beta}$ can be obtained by solving the generalized estimating equations:
\begin{equation}
\sum_{i=1}^N \textbf{D}_i'\textbf{V}_i^{-1}(\mathbf{y}_i-\bm{\mu}_i)=\mathbf{0}, \nonumber
\end{equation}

\noindent
where $\textbf{D}_i=\partial\bm{\mu}_i/\partial\bm{\beta}$ is a matrix of order $n_i\times p$. Given the method of moment estimates of $\bm{\alpha}$ and $\phi$, the iterative Fisher-Scoring (F-S) algorithm for estimating $\boldsymbol{\beta}$ can be defined as follows:

\begin{equation}
\boldsymbol{\beta}_{r+1}=\boldsymbol{\beta}_{r}+\Big(\sum\limits_{i=1}^N \mathbf{D}_{i}^{'}\textbf{V}_{i}^{-1}\mathbf{D}_{i}\Big)^{-1} \Big(\sum\limits_{i=1}^N \mathbf{D}_{i}^{'}\textbf{V}_{i}^{-1}(\mathbf{y}_{i}-\boldsymbol{\mu}_{i})\Big), \nonumber
\end{equation} 

\noindent
where $r$ is the iteration number. Note that $\bm{\alpha}$ and $\phi$ are estimated through the method of moments using Pearson residuals repeatedly within the iterative F-S algorithm. For large $N$, the estimator $\hat{\bm{\beta}}$ follows an asymptotic $p$-variate normal distribution with mean vector $\bm{\beta}$ and covariance matrix: 

\begin{equation}
\text{Cov}(\hat{\bm{\beta}})=\textbf{M}^{-1}\textbf{J}\textbf{M}^{-1}, \nonumber
\end{equation} 

\noindent
where $\textbf{M}=\sum_{i=1}^N \textbf{M}_i$ with $\textbf{M}_i=\textbf{D}_i'\textbf{V}_i^{-1}\textbf{D}_i$ and  $\textbf{J}=\sum_{i=1}^N \textbf{D}_i'\textbf{V}_i^{-1}\text{Cov}(\mathbf{y}_i)\textbf{V}^{-1}_i \textbf{D}_i$.  The sandwich (which is also named as robust or empirical) covariance estimator for $\hat{\bm{\beta}}$ is obtained by replacing $\text{Cov}(\mathbf{y}_i)$ by $(\mathbf{y}_i-\bm{\mu}_i)(\mathbf{y}_i-\bm{\mu}_i)^{'}$ in $\textbf{J}$ such that:

\begin{equation}
\bm{\Sigma}=\textbf{M}^{-1}\bigg(\sum_{i=1}^N \textbf{D}_i'\textbf{V}_i^{-1}(\mathbf{y}_i-\bm{\mu}_i)(\mathbf{y}_i-\bm{\mu}_i)^{'}\textbf{V}^{-1}_i \textbf{D}_i \bigg)\textbf{M}^{-1}. \nonumber
\end{equation}

\noindent
It is well-known that the consistency and asymptotic normality of 
$\hat{\bm{\beta}}$ hold even if $\textbf{R}_i$ is mis-specified. However, when $\textbf{R}_i$ is correctly specified, then $\text{Cov}(\mathbf{y}_i)=\textbf{V}_i$ and $\bm{\Sigma}$ reduces to the model-based (which is also named as naive) covariance estimator {$\textbf{M}^{-1}$}. 

\subsection{Overview of existing GEE working correlation structure selection criteria}

The quasi-likelihood information criterion (QIC) \citep{pan2001akaike} can be defined as follows:

\begin{equation}
QIC=-2 QL(\hat{\bm{\beta}},\hat{\phi}) + 2 trace(\hat{\textbf{M}}_{I}\hat{\bm{\Sigma}}), \nonumber
\end{equation}

\noindent
where $QL(\hat{\bm{\beta}},\hat{\phi})$ is the quasi-likelihood function (e.g., it is $\hat{\phi}\sum_{i}^{N}\big(\textbf{y}_{i}\log(\hat{\bm{\mu}}_{i}/(1-\hat{\bm{\mu}}_{i})+\log(1-\hat{\bm{\mu}}_{i})\big)$ and $\hat{\phi}\sum_{i}^{N}\big(\textbf{y}_{i}\log(\hat{\bm{\mu}}_{i})-\hat{\bm{\mu}}_{i}\big)$ for binary and Poisson outcomes, respectively) and $\hat{\textbf{M}}_{I}=\sum_{i=1}^N\hat{\textbf{D}}_i'\hat{\textbf{A}}_i^{-1}\hat{\textbf{D}}_i$. The terms $QL(\hat{\bm{\beta}},\hat{\phi})$, $\hat{\textbf{M}}_{I}$, and
$\hat{\bm{\Sigma}}$ are evaluated at the values of 
$\hat{\bm{\beta}}$, $\hat{\bm{\alpha}}$, and $\hat{\phi}$ under the assumed working correlation matrix.

The correlation information criterion (CIC) \citep{hin2009working}, which uses the trace term of QIC only,  can be defined as follows:

\begin{equation}
CIC=trace(\hat{\textbf{M}}_{I}\hat{\bm{\Sigma}}), \nonumber
\end{equation}

\noindent
where $\hat{\textbf{M}}_{I}$ and
$\hat{\bm{\Sigma}}$ are evaluated at the values of 
$\hat{\bm{\beta}}$, $\hat{\bm{\alpha}}$, and $\hat{\phi}$ under the assumed working correlation matrix.

On the other hand, \citet{rotnitzky1990hypothesis} used the expression {$\textbf{Q}=\hat{\textbf{M}}\hat{\bm{\Sigma}}$} to define three criteria for selecting the true correlation structure:

\begin{eqnarray}
\begin{split}
\text{RJ1} &=trace(\textbf{Q})/p\\ \nonumber
\text{RJ2} &=trace(\textbf{Q}^2)/p, \quad \text{and} \\ \nonumber
\text{DBAR}&=\text{RJ2-2RJ1+1}, \nonumber
\end{split}
\end{eqnarray}

\noindent
where $p$ is the dimension of the regression parameter $\bm{\beta}$. The terms $\hat{\textbf{M}}$ and $\hat{\bm{\Sigma}}$ correspond to
$\textbf{M}$ and $\bm{\Sigma}$ evaluated at the values of $\hat{\bm{\beta}}$, $\hat{\bm{\alpha}}$, and $\hat{\phi}$ under the assumed working correlation matrix. 

Finally, the {SC} criterion \citep{shults1998analysis} can be expressed as follows:

\begin{equation}
\label{eq:sc}
SC=\sum_{i=1}^N (\mathbf{y}_i-\hat{\bm{\mu}}_i)^{'}\hat{\textbf{V}}_i^{-1}(\mathbf{y}_i-\hat{\bm{\mu}}_i), 
\end{equation}

\noindent
where $\hat{\bm{\mu}}_i$ and
$\hat{\textbf{V}}_i$ are evaluated at the values of 
$\hat{\bm{\beta}}$, $\hat{\bm{\alpha}}$, and $\hat{\phi}$ under the assumed working correlation matrix. Working correlation structure selection is based on the premise that smaller values of {QIC}, {CIC}, {SC}, {$\mid\text{RJ1-1}\mid$}, {$\mid\text{RJ2}-1\mid$}, and $\mid\text{DBAR}\mid$ lead to a better working correlation structure. 

\section{A PRESS statistic for working correlation selection in GEE}
\label{sec:press}

A second class of working correlation structure selection criterion in GEE analysis is based on the weighted error sum of squares $\sum_i\mathbf{e}_{i}'\textbf{W}_i\mathbf{e}_{i}$, where $\mathbf{e}_{i}=\mathbf{y}_i-\bm{\mu}_i$ and $\textbf{W}_i$ is a weight matrix.  \citet{shults1998analysis} defined $\textbf{W}_i=\textbf{V}_i^{-1}$ to propose the {SC} criterion in equation~\ref{eq:sc}, hypothesizing that the correct working correlation structure should minimize the error sums of squares weighted by the inverse of the working covariance structure. However, the reason for its poor performance relative to other criteria \citep{shults2009comparison} is because $\textbf{V}_i$  is a biased estimator of $\text{Cov}(\mathbf{e}_i)$; a better estimator is $\text{Cov}(\mathbf{e}_i)\approx (\textbf{I}-\textbf{H}_i)\textbf{V}_i(\textbf{I}-\textbf{H}_i')$,
where $\textbf{H}_i=\textbf{D}_i\textbf{M}^{-1}\textbf{D}_i'\textbf{V}_i^{-1}$.  The cluster-level leverage matrix $\textbf{H}_i$ has an important role in GEE inference appearing in formulae for bias-corrected covariance estimators for $\hat{\bm{\beta}}$ \citep{mancl2001covariance}, where {$(\textbf{I} - \textbf{H}_i)^{-1}\mathbf{e}_{i}\mathbf{e'}_{i}(\textbf{I} - \textbf{H}'_i)^{-1}$} replaces $\mathbf{e}_{i}\mathbf{e'}_{i}=(\mathbf{y}_i-\bm{\mu}_i)(\mathbf{y}_i-\bm{\mu}_i)^{'}$ in $\boldsymbol \Sigma$.  Furthermore, $\textbf{H}_i$ appears in GEE regression diagnostics that approximate the change in $\hat{\bm{\beta}}$ due to cluster-deletion \citep{hammill2006sas,preisser1996deletion}.

Using the weight matrix $\textbf{W}_i=(\textbf{I}-\textbf{H}_i')^{-1}\textbf{V}_i^{-1}(\textbf{I}-\textbf{H}_i)^{-1}$, a GEE PRESS Criterion (GPC) for working correlation structure selection, which approximates $\sum_i^{N}\hat{\mathbf{e}}_i'\{\widehat{\text{Cov}}(\hat{\mathbf{e}}_i)\}^{-1}\hat{\mathbf{e}}_i$, is given as follows:
\begin{align}\label{lpstat}
\text{GPC} &= \sum_{i=1}^N \hat{\mathbf{e}}_i'(\textbf{I}-\hat{\textbf{H}}_i')^{-1}\hat{\textbf{V}}_i^{-1}(\textbf{I}-\hat{\textbf{H}}_i)^{-1}\hat{\mathbf{e}}_i,
\end{align}
where $\hat{\mathbf{e}}_i=(\mathbf{y}_i-\hat{\bm{\mu}}_i)$, and $\hat{\bm{\mu}}_i$, $\hat{\textbf{H}}_i$, and $\hat{\textbf{V}}_i$ are evaluated at the values of $\hat{\bm{\beta}}$, $\hat{\bm{\alpha}}$, and $\hat{\phi}$ under the assumed working correlation matrix.

Appendix A provides further rationale for GPC as a working correlation selection criterion by showing that it also approximates $\sum_i^{N} \hat{\mathbf{e}}_{(i)}' \{\widehat{\text{Cov}}(\hat{\mathbf{e}}_{(i)})\}^{-1} \hat{\mathbf{e}}_{(i)}$, where
$\hat{\mathbf{e}}_{(i)}=\mathbf{y}_i-g^{-1}(\textbf{X}_i\hat{\bm{\beta}}_{(i)})$ is the corrected {PRESS} residual vector corresponding to the $i$th cluster with $\hat{\bm{\beta}}_{(i)}$ being the estimate of ${\bm{\beta}}$ that is estimated without the $i$th cluster, i.e.\ with $(N-1)$ clusters/subjects.  The GPC shares the attractive feature of GEE cluster-deletion diagnostics DBETA (defined in the appendix) in that they are computationally fast formula not requiring further iteration beyond convergence of the usual iterative reweighted least squares algorithm \citep{preisser2008note}. For independent observations (equivalently, $n_i=1$ for all $i$), the identity link (i.e., $L_i = 1$), and constant variance (i.e., $h(\mu_i)=1$), GPC reduces to the PRESS statistic for multiple linear regression $\sum_{i=1}^N (y_i - \hat{y}_{(i)})^2 = \sum_{i=1}^N e_i^2/(1-h_i)^2,$ where $h_i = \mathbf{X}'_i(\textbf{X}_i'\textbf{X}_i)^{-1}\mathbf{X}_i.$

\section{Simulation study}
\label{sec:sim}

In this section, we carry out an extensive simulation study to compare performance of the proposed {GPC} with several working correlation selection criteria ({CIC}, {DBAR}, {QIC}, {RJ1}, {RJ2}, and {SC}) in terms of proportion of selecting the true working correlation structure in longitudinal data under different scenarios and we discuss the results of the simulation study.

\subsection{Simulation design}

Specifically, we evaluate the criteria in four main simulation scenarios:
 
\begin{itemize}
\item[S1.] Longitudinal binary data with equal number of time points,
\item[S2.] Longitudinal binary data with unequal number of time points,
\item[S3.] Longitudinal count data with equal number of time points, and
\item[S4.] Longitudinal count data with unequal number of time points.
\end{itemize}

For scenarios $1$ and $2$, we assume that the response $y_{it}$, corresponding to the $t$th observation of the $i$th subject, follows a Bernoulli distribution with marginal mean $\mu_{it}$, which depends on the covariates $x_1$ and $x_2$ via a  logit link function according to the following model:
\begin{align}
\label{eq:binary}
\log \bigg(\frac{\mu_{it}}{1-\mu_{it}}\bigg) &= \beta_0 + \beta_1x_{1i} + \beta_2 x_{2it},\;\;i=1,\ldots,N,\;\;t=1,\ldots, n_{i},
\end{align}
where the covariates $x_{1i}$ and $x_{2it}$ are both binary, generated from $\text{Bernoulli}(0.5)$. In defining the mean function, $x_{1}$ is a subject-specific covariate, i.e.,\ it takes the same value for all the observations within a subject, and $x_2$ is a {time}-specific covariate, i.e.,\ it takes different values over the observations of a subject. In other words, $x_1$ and $x_2$ can be considered as time-independent and time-dependent covariates, respectively. True values of the regression parameters used in the simulation are $\beta_0=1$, $\beta_1=0.38$, and $\beta_2=0.35$. The true values of the regression coefficients are selected so that proportion of ones in the response variable is around $50\%$.

For scenarios $3$ and $4$, we assume that the response $y_{it}$, corresponding to the $t$th observation of the $i$th subject, follows a Poisson distribution with marginal mean $\mu_{it}$, which depends on the covariates $x_1$ and $x_2$ via a $\log$ link function according to the following model:
\begin{align}
\label{eq:Poisson}
\log \big(\mu_{it}\big) &= \beta_0 + \beta_1x_{1i} + \beta_2 x_{2it},\;\;i=1,\ldots,N,\;\;t=1,\ldots, n_{i},
\end{align}
where the covariates $x_{1i}$ and $x_{2it}$ are as defined in equation~\ref{eq:binary}. True values of the regression parameters used in the simulation are $\beta_0=1$, $\beta_1=0.20$, and $\beta_2=0.40$. The true values of the regression coefficients are selected so that
overall mean of counts is $3.5$. 

Under each scenario, two different sample sizes ($N=50$ and $100$) are investigated. For the scenarios with equal number of time points, $n_{i}$ is equal to $5$ for all subjects. For the scenarios with unequal number of time points, $n_{i}$ for each subject randomly takes a value from the set $\{3,4,5\}$ with corresponding probabilities $\{0.15,0.15,0.7\}$. This approach assumes that every subject has measurements observed at first three time points (e.g., $1$, $2$, and $3$) and $70\%$ of the subjects are fully observed at five time points,
$15\%$ of the subjects are fully observed at first four time points, and
$15\%$ of the subjects are fully observed at first three time points.
Furthermore, for each scenario, three different true working correlation structures, AR(1), Exch, and UN are investigated with two different degree of within-subject correlation values ($\alpha=0.2$ and $0.4$). For the UN correlation structure, the correlation between the observations at $t$ and $t'$ time points is defined as $\alpha^{{\mid t-t'\mid}^{\lambda}}$, which reduces to AR(1) and Exch structures for $\lambda=1$ and $\lambda=0$, respectively, and $\lambda=0.5$ is used for the UN structure. The over-dispersion parameter $\phi$ is set to $1$ for each case. Under each of $48$ different scenarios, $1000$ longitudinal replicate data sets are generated via {\verb"R"} \citep{rcore} package {\verb"PoisBinNonNor"} \citep{PoisBinNonNor}. Note that {\verb"R"} package {\verb"PoisBinNonNor"} also allows to check if there are range violations among pair-wise correlations of binary-binary variables and Poisson-Poisson variables given the true value of the marginal mean and correlation based on the methodology proposed in \citet{demirtas2011practical}.

For each of the true correlation structures considered (i.e.\ AR(1), Exch, and UN), the method of GEE is fitted with four different working correlation structures, namely, independence (Indep), AR(1), Exch, and UN via assuming there is no over-dispersion $\phi=1$ in the outcomes. User defined {R} codes, which can handle unequally spaced and unbalanced longitudinal data (automatically covers equally spaced unbalanced longitudinal data and balanced longitudinal data), are used to fit all the models considered in the simulation.

\subsection{Simulation results}

Under each of $48$ different scenarios, the proportion of times the true working correlation structure is selected by  {CIC}, {DBAR}, {GPC}, {QIC}, {RJ1}, {RJ2}, and {SC} criteria among the candidate working correlation structures Indep, AR(1), Exch, and UN are summarized in Tables \ref{tab1M1}-\ref{tab1M8}. The bold values in Tables \ref{tab1M1}-\ref{tab1M8} show the highest score within that true correlation structure and sample size combination. Furthermore, the results for the mean squared error values of parameters for each scenario are presented in Appendix B.

Overall, the pattern of results can be described according to three groups among 
the seven working correlation selection criteria defined in terms of their functional (mathematical) form: $\{CIC, QIC\}$, $\{DBAR,RJ1,RJ2\}$, and $\{GPC, SC\}$, respectively. The first main result is that {CIC} always selects the {UN} correlation structure as the best structure in all $48$ scenarios, with a value greater than $80\%$, where {QIC} does so $40\%-65\%$ of the time. Hence, regardless of the true correlation structure, these two criteria are likely to pick UN as the best structure which calls for caution when using them.

The other main result is that the performance of {GPC} is always better than that of {SC} in all $48$ scenarios indicating that use of the weight matrix $\textbf{W}_i=(\textbf{I}-\textbf{H}_i')^{-1}\textbf{V}_i^{-1}(\textbf{I}-\textbf{H}_i)^{-1}$ results in a better performance than use of $\textbf{V}_i^{-1}$.

Notably, when the true working correlation structure is AR(1), {GPC} always gives the highest score among all criteria (Tables \ref{tab1M1}-\ref{tab1M8}). When the within-subject association level increases from $\alpha=0.2$ to $\alpha=0.4$, proportion of correct selections by GPC increases (e.g., in  Table \ref{tab1M1} under AR(1) true working correlation structure and $N=50$, GPC has a value of $0.364$, whereas in  Table~\ref{tab1M2} under the same settings it is value is $0.421$). The performance of GPC also depends on the response type. GPC has slightly better performance for Poisson responses compared to binary responses (e.g., in  Table~\ref{tab1M4} under AR(1) true working correlation structure and $N=50$, GPC has a value of $0.407$, whereas in  Table~\ref{tab1M8} under the same settings it is value is $0.452$).

With respect to identifying the true Exch structure, GPC performs nearly as well as DBAR, RJ1, and RJ2 criteria. In fact, the performance of GPC exceeds $60\%$ in scenarios with Poisson responses, whereas the performance of RJ2 is around $70\%$ in the same cases (e.g., please see Tables \ref{tab1M5} and \ref{tab1M7}).

When the true working correlation structure is UN, drawing conclusions from the simulation results is not straightforward. As mentioned earlier, regardless of the true structure, the criteria {CIC} and {QIC} are very likely to select UN as the best structure. On the other hand, the RJ criteria more often correctly select the unstructured correlation matrix than GPC does. Indeed, the performance of GPC is between $5\%-25\%$. Actually, the results in Tables~\ref{tab1S1}-\ref{tab1S8} in the Appendix also show that even if the true correlation structure is UN, fitting a marginal model via GEE with an UN structure results in higher mean squared error values compared to a model with a simple correlation structure. This implies that the most preferred criteria are the ones with a tendency to pick correlation types with simple structures. This is because the inflation in the number of correlation parameters to be estimated in a model results in a loss of efficiency in regression parameter estimates. 


At this point we would like to point that authors of similar studies on selection of working correlation structure in marginal analysis of longitudinal data via GEE
either do not consider the UN structure as a candidate structure \citep{pardo2017working} or they specify a large number of candidate structures so that the influence of the UN structure on the simulation results is decreased \citep{jamanworking2015}. 

To assess the effect of omitting UN as a candidate working correlation structure, we present limited simulation results showing the proportion of times the criteria select the true correlation structure among Indep, AR(1), and Exch (excluding UN) structures, for balanced longitudinal binary data with different sample sizes when the true within-subject correlation level is $\alpha=0.2$. The results in Table \ref{tab1M9} show that {CIC} is very good at selecting the true correlation structure and GPC performs nearly as well as CIC (and better in the Exch and $N=50$ scenario) while incorrectly choosing the Indep structure less often than CIC.

\section{Illustration of working correlation structure selection with the coronary artery risk development in young adults (CARDIA) study}
\label{sec:example}

In this section, binary outcomes for smoking status in a longitudinal cohort of $5,077$ subjects in the Coronary Artery Risk Development in Young Adults (CARDIA) study are considered to illustrate the choice of working correlation structure selection in marginal analysis of longitudinal data with GEE. In the CARDIA study, the cigarette smoking status of participants (yes, no) was ascertained at six study visits beginning in $1986$ (year $0$) and continuing in follow-up years $2$, $5$, $7$, $10$, and $15$ (resulting in unequally spaced longitudinal binary data). The age at baseline (ranges between $18-30$ years), attained education status by the end of study (no college attended, some college, college degree), sex (female, male), and race (black, white) information were also gathered from each participant. 

\citet{perin2009wgee} analyzed the CARDIA study data with GEE where smoking status was the response variable with follow-up year treated as a categorical variable and a working correlation structure with an Exch type was assumed. The GEE analysis of these data indicated that smoking rates significantly declined from baseline (year $0$) to year $15$ with greater declines for females than males and for whites compared to blacks. The choice of working correlation structure was based on exploratory analysis rather than any objective criterion. In particular, the Exch type of correlation structure appeared to adequately characterize the within-subject correlation of smoking status over time for CARDIA study participants, where smoking initiation and quitting occur infrequently.


In this study, we re-visit the $15$-year CARDIA study data using cubic polynomial models for time with adjustment for participant age, age-squared, and attained education status. Several GEE with different working correlation structures for smoking status (yes, no) with a logit link function are applied to the longitudinal binary data on four subgroups of sex and race, i.e., black females ($N=1473, \bar{n}= 4.90$), black males ($N=1145, \bar{n}= 4.66$), white females ($N=1299, \bar{n}= 5.29$), and white males ($N=1160, \bar{n}= 5.29$), respectively, with $N$ denoting the number of participants in the subgroup and $\bar{n}$ denoting the average number of visits per participant in the subgroup (i.e., where $n_{i}$ is between $1-6$ for each participant). 

Specifically, let the response $y_{it}=1$ if the $i$th young adult at the $t$th time point was a smoker and 0 otherwise and $\mu_{it}=E[y_{it}]$, the probability that the $i$th young adult at the $j$th time point was a smoker ($i=1,\ldots,N$, $t=1,\ldots, n_i$), can be described by the following model:

\begin{equation}
\label{eq:model2}
\log \bigg(\frac{\mu_{it}}{1-\mu_{it}}\bigg) = \beta_0 + \beta_1x_{1i} + \beta_2 x_{2it}
+ \beta_3x_{3i} + \beta_4 x_{4i}
+ \beta_L year_{t} + \beta_Q year^2_{t} 
+  \beta_Cyear^3_{t},
\end{equation}

\noindent where $x_{1i}$ = age in years divided by 10, $x_{2i}=x_{1i}^2$, $x_{3i} = 1$ if attained education is some college without degree and 0 otherwise, $x_{4i} = 1$ if attained education is college degree and 0 otherwise, and $year_{t}$ is years since the first study exam in 1986 divided by 10.  As a primary question of interest, we investigate whether smoking rates among young adults changed over time.  Secondarily, we examine whether the smoking rate at the first study exam ($year_{1}=0$) differs from the smoking rate fifteen years later ($year_{6}=1.5$ decades).

Table~\ref{tab:tab2M3} shows the values of the various working correlation selection criteria for each race and sex group. The GPC  identifies the AR(1) working correlation as the best structure for black females, black males, and white males groups and Exch working correlation as the best structure for white females group. SC also agrees with GPC on all cases. Conversely, the CIC, DBAR, RJ1, and RJ2 criteria almost always identify the unstructured correlation as the best structure for all groups.   

Application of GEE with the correlation structures suggested by GPC criterion are presented in Table \ref{tab3}. The results indicate a strong effect of education where smoking rates are significantly greatest among those without college attendance compared to those with at least some college and those with a college degree. For example, while the odds of smoking for a black male without college attendance is $1/exp(-0.964)=2.622$ times higher than that for a black male with a some college attendance, it is $1/exp(-1.787)= 5.972$ times higher compared to that for a black male with a college degree. Based on the results in Table \ref{tab4}, the empirical score test for trend is significant for every group (e.g., P-values are less than $0.001$, $0.05$, $0.001$, and $0.001$, respectively) suggesting that the smoking rate changes over time. However, when the year 1986 (year $0$) smoking rate is compared to the year 2001 (year $1.5$) smoking rate, the difference in rates is not statistically significant for black males (e.g., P-value is $0.580$).  

Fig.~\ref{fig:fig1} shows the smoking trends based on the model applied to the four race and sex subgroups separately when using an AR(1) correlation structure for black females, black males, and white males and an Exch correlation for white females. Fig.~\ref{fig:fig1} shows model-adjusted trends in smoking for each of the three education groups for each race/sex group of young adults and indicates that smoking rates are higher in participants with high school education or less and is much more pronounced in black participants compared to white participants.  

\section{Conclusion}
\label{sec:conclusion}

This study proposed and evaluated a PRESS statistic for selection of the working correlation structure in marginal analysis of longitudinal data by generalized estimating equations. The simulation results showed that the proposed PRESS statistic has better performance than the SC criterion in all cases, is better than several other existing criteria when the true working correlation structure is AR(1), and is comparable with other criteria when the true working correlation structure is exchangeable. However, the simulation studies showed that GPC performs somewhat poorly with respect to identifying an unstructured correlation structure. This is not necessarily an undesirable result as marginal models with a UN structure fitted with GEE results in higher mean squared error compared to a model with a simpler correlation structure. 

Correctly discriminating between an exchangeable (or AR(1)) versus UN correlation matrices is challenging owing to the number of correlation parameters to be estimated in the UN case. Indeed, \citet{Westgate2013,Westgate2016} showed the estimation of a large number of correlation parameters in the UN structure leads to an increase in the sampling variances of the regression parameters. \citet{Westgate2014} proposed a penalization approach to the estimation of the sandwich covariance matrix of the regression parameters in the UN structure which results in higher selection accuracy. Westgate's penalty is applicable to any working correlation selection criteria involving the sandwich covariance estimator of the regression parameters. Since our proposed criterion does not involve the sandwich variance estimator of the regression parameters, the penalty cannot be applied to the GPC criterion. Nonetheless, the GPC criterion needs further development to result in better accuracy in selection of an UN correlation structure.

We focused on selection of the working correlation matrix under the assumption that the marginal mean model is correctly specified. Future investigations of the GPC statistic may include its evaluation for selection of covariates in the linear predictor when correct selection of the correlation matrix is assumed in a GEE analysis. A special case of the GPC statistic for independent responses warrants investigation for variable selection in generalized linear models. For the regression analysis of clustered data with marginal models using GEE, a generalized version of Mallow's $C_p$ was shown to perform well relative to variable selection based on Wald and score tests \citep{cantoni2005GEE}. Other problems that could be addressed include the simultaneous selection of working covariance matrix and the linear predictor in GEE.

As a final note, we should also note that {\verb"R"} codes used for fitting marginal models through GEE in this study are adapted from the codes of the {\verb"R"} package {\verb"PGEE"} \citep{PGEE} which implements penalized generalized estimating equations for analysis of longitudinal data with large number of covariates. The {\verb"R"} codes used in this study are extended for the analysis of unequally spaced unbalanced longitudinal data via the method of GEE. 

Finally, we would like to draw the attention of authors using {\verb"R"} package {\verb"geepack"} \citep{geepack}, {\verb"R"} package {\verb"gee"} \citep{gee}, and {\verb"SAS"} {\verb"GENMOD"} procedure \citep{sas} in working correlation structure selection problems in GEE to the following remarks:  1) Our small-scale simulation studies show that {\verb"R"} package {\verb"geepack"}, which uses another set of estimating equations for correlation parameters, estimates the correlation parameter under AR(1) correlation structure with a considerable bias; 2) The {\verb"R"} package {\verb"gee"} reports an error under AR(1) correlation structure when one of the subjects has only one time point, although it is possible to estimate the correlation parameter; 3) The {\verb"R"} package {\verb"gee"} cannot handle unequally spaced longitudinal data (whether it is balanced or unbalanced) under AR(1) and UN correlation structures since there is not any input argument taking the spacing between time points as in the {waves} argument in {\verb"R"} package {\verb"geepack"} and {withinsubject} option in {repeated} statement of {\verb"SAS"} {\verb"GENMOD"} procedure; and 4) The {\verb"R"} package {\verb"gee"} and {\verb"SAS"} {\verb"GENMOD"} procedure estimate the over-dispersion parameter, $\phi$, even if it is set to $1$. 


\section*{Acknowledgment}
This study was presented at the 38th Annual Conference of the International Society for Clinical Biostatistics (ISCB) in Vigo, Spain in 2017 through a travel grant programme of the ISCB conference.

\newpage
\clearpage
\section*{Tables and Figures}

\begin{table}[ht]
\centering
\captionsetup{width=0.75\textwidth}
\caption{Proportion of selecting the true correlation structure by the competing selection criteria for balanced longitudinal binary data with different sample sizes when true within-subject correlation level is $\alpha=0.2$.}
\smallskip
\label{tab1M1}
\scalebox{0.7}{
\begin{tabular}{@{}llcccccccccc@{}} \hline
True&&\multicolumn{4}{c}{$N=50$} && \multicolumn{4}{c}{$N=100$}\\
\cmidrule{3-6}\cmidrule{8-11}
correlation &Selection&\multicolumn{4}{c}{Working corr. structure}&& \multicolumn{4}{c}{Working corr. structure}\\
\cmidrule{3-6}\cmidrule{8-11}
structure &criteria           &{Indep} & {AR(1)} & {Exch}  &  {UN}  && {Indep} & {AR(1)} & {Exch}  &  {UN} \\
\midrule
            &$\text{CIC}$       & 0.028 & 0.089 & 0.028 & 0.859 && 0.017 & 0.062 & 0.025 & 0.90\\
            &$\text{DBAR}$      & 0.030 & 0.181 & 0.366 & 0.474 && 0.006 & 0.165 & 0.421 & 0.500 \\
            &$\text{GPC}$        & 0.040 & \textbf{0.364} & 0.397 & 0.199 && 0.031 & \textbf{0.324} & 0.417 & 0.228 \\
{AR(1)}     &$\text{QIC}$       & 0.115 & 0.238 & 0.108 & 0.547 && 0.105 & 0.225 & 0.105 & 0.570  \\
            &$|\text{RJ1}-1|$   & 0.087 & 0.240 & 0.393 & 0.309 && 0.051 & 0.269 & 0.343 & 0.361\\
            &$|\text{RJ2}-1|$   & 0.050 & 0.270 & 0.391 & 0.295 && 0.027 & 0.277 & 0.343 & 0.363\\
            &$\text{SC}$        & 0.259 & 0.327 & 0.213 & 0.203 && 0.268 & 0.298 & 0.209 & 0.227\\ \hline
                    
          &$\text{CIC}$         & 0.018 & 0.038 & 0.076 & 0.868 && 0.013 & 0.026 & 0.043 & 0.924\\
          &$\text{DBAR}$        & 0.002 & 0.026 & 0.541 & 0.452 && 0.000 & 0.002 & 0.555 & 0.496  \\    
          &$\text{GPC}$          & 0.028 & 0.218 & 0.582 & 0.172 && 0.038 & 0.314 & 0.489 & 0.159\\
{Exch}    &$\text{QIC}$         & 0.156 & 0.109 & 0.218 & 0.519 && 0.126 & 0.076 & 0.212 & 0.588 \\
          &$|\text{RJ1}-1|$     & 0.014 & 0.099 & 0.535 & 0.360 && 0.000 & 0.028 & 0.544 & 0.450\\
          &$|\text{RJ2}-1|$     & 0.007 & 0.048 & \textbf{0.585} & 0.372 && 0.000 & 0.010 & \textbf{0.566} & 0.432\\ 
          &$\text{SC}$          & 0.305 & 0.335 & 0.209 & 0.152 && 0.250 & 0.386 & 0.228 & 0.136  \\ \hline 
                     
          &$\text{CIC}$         & 0.029 & 0.077 & 0.042 & \textbf{0.856} && 0.017 & 0.050 & 0.017 & \textbf{0.918}   \\ 
          &$\text{DBAR}$        & 0.008 & 0.120 & 0.443 & 0.471 && 0.003 & 0.089 & 0.492 & 0.481   \\      
          &$\text{GPC}$          & 0.034 & 0.342 & 0.447 & 0.178 && 0.046 & 0.327 & 0.411 & 0.216 \\
{UN}      &$\text{QIC}$         & 0.152 & 0.188 & 0.138 & 0.526 && 0.121 & 0.204 & 0.125 & 0.556 \\
          &$|\text{RJ1}-1|$     & 0.049 & 0.258 & 0.398 & 0.308 && 0.011 & 0.231 & 0.411 & 0.375    \\
          &$|\text{RJ2}-1|$     & 0.025 & 0.222 & 0.467 & 0.296 && 0.007 & 0.171 & 0.454 & 0.376     \\
          &$\text{SC}$          & 0.289 & 0.348 & 0.192 & 0.172 && 0.274 & 0.320 & 0.195 & 0.212\\ \hline                      
\end{tabular}}
\end{table}

\begin{table}[ht]
\centering
\captionsetup{width=0.75\textwidth}
\caption{Proportion of selecting the true correlation structure by the competing selection criteria for balanced longitudinal binary data with different sample sizes when true within-subject correlation level is $\alpha=0.4$.}
\smallskip
\label{tab1M2}
\scalebox{0.7}{
\begin{tabular}{@{}llcccccccccc@{}} \hline
True&&\multicolumn{4}{c}{$N=50$} && \multicolumn{4}{c}{$N=100$}\\
\cline{3-6}\cline{8-11}
correlation &Selection&\multicolumn{4}{c}{Working corr. structure}&& \multicolumn{4}{c}{Working corr. structure}\\
\cline{3-6}\cline{8-11}
structure &criteria    &{Indep} & {AR(1)} & {Exch}  &  {UN}  && {Indep} & {AR(1)} & {Exch}  &  {UN} \\
\hline
                    &$\text{CIC}$       & 0.009 & 0.209 & 0.027 & 0.757 && 0.002 & 0.134 & 0.007 & 0.861 \\ 
                    &$\text{DBAR}$      & 0.002 & 0.234 & 0.471 & 0.317 && 0.000 & 0.205 & 0.468 & 0.375   \\      
                    &$\text{GPC}$        & 0.030 & \textbf{0.421} & 0.448 & 0.102 && 0.039 & \textbf{0.412} & 0.446 & 0.104 \\
{AR(1)}             &$\text{QIC}$       & 0.105 & 0.330 & 0.111 & 0.457 && 0.068 & 0.320 & 0.119 & 0.495\\
                    &$|\text{RJ1}-1|$   & 0.009 & 0.299 & 0.452 & 0.250 && 0.000 & 0.274 & 0.414 & 0.332    \\
                    &$|\text{RJ2}-1|$   & 0.001 & 0.292 & 0.469 & 0.245 && 0.000 & 0.278 & 0.437 & 0.295     \\
                    &$\text{SC}$        & 0.318 & 0.387 & 0.195 & 0.100 && 0.284 & 0.389 & 0.223 & 0.104\\ \hline 
                    
                    &$\text{CIC}$       & 0.021 & 0.053 & 0.180 & 0.750 && 0.003 & 0.019 & 0.105 & 0.874    \\ 
                    &$\text{DBAR}$      & 0.000 & 0.011 & 0.663 & 0.338 && 0.000 & 0.000 & \textbf{0.651} & 0.386    \\      
                    &$\text{GPC}$        & 0.041 & 0.309 & 0.558 & 0.092 && 0.057 & 0.366 & 0.515 & 0.062\\
{Exch}              &$\text{QIC}$       & 0.150 & 0.086 & 0.331 & 0.438 && 0.158 & 0.044 & 0.327 & 0.475\\
                    &$|\text{RJ1}-1|$   & 0.000 & 0.058 & 0.623 & 0.320 && 0.000 & 0.011 & 0.601 & 0.405    \\
                    &$|\text{RJ2}-1|$   & 0.000 & 0.013 & \textbf{0.703} & 0.285 && 0.000& 0.001 & 0.640 & 0.361    \\
                    &$\text{SC}$        & 0.356 & 0.343 & 0.219 & 0.082 && 0.361 & 0.361 & 0.228 & 0.050\\ \hline 
                   
                    &$\text{CIC}$       & 0.016 & 0.140 & 0.071 & \textbf{0.776} && 0.002 & 0.101 & 0.038 & \textbf{0.861}   \\ 
                    &$\text{DBAR}$      & 0.000 & 0.094 & 0.583 & 0.342 && 0.000 & 0.049 & 0.597 & 0.387   \\      
                    &$\text{GPC}$        & 0.025 & 0.385 & 0.500 & 0.090 && 0.055 & 0.374 & 0.501 & 0.070 \\
{UN}                &$\text{QIC}$       & 0.130 & 0.232 & 0.181 & 0.460 && 0.094 & 0.210 & 0.203 & 0.493 \\
                    &$|\text{RJ1}-1|$   & 0.001 & 0.215 & 0.513 & 0.279 && 0.000 & 0.145 & 0.526 & 0.344   \\
                    &$|\text{RJ2}-1|$   & 0.000 & 0.144 & 0.599 & 0.259 && 0.000 & 0.099 & 0.585 & 0.326    \\
                    &$\text{SC}$        & 0.322 & 0.379 & 0.217 & 0.082 && 0.334 & 0.368 & 0.238 & 0.060\\ \hline                      
\end{tabular}}
\end{table}

\begin{table}[ht]
\centering
\captionsetup{width=0.75\textwidth}
\caption{Proportion of selecting the true correlation structure by the competing selection criteria for unbalanced longitudinal binary data with different sample sizes when true within-subject correlation level is $\alpha=0.2$.}
\smallskip
\label{tab1M3}
\scalebox{0.7}{
\begin{tabular}{@{}llcccccccccc@{}} \hline
True&&\multicolumn{4}{c}{$N=50$} && \multicolumn{4}{c}{$N=100$}\\
\cmidrule{3-6}\cmidrule{8-11}
correlation &Selection&\multicolumn{4}{c}{Working corr. structure}&& \multicolumn{4}{c}{Working corr. structure}\\
\cmidrule{3-6}\cmidrule{8-11}
structure &criteria    &{Indep} & {AR(1)} & {Exch}  &  {UN}  && {Indep} & {AR(1)} & {Exch}  &  {UN} \\
\midrule
                    &$\text{CIC}$       & 0.040 & 0.095 & 0.043 & 0.827 && 0.012 & 0.069 & 0.017 & 0.904  \\ 
                    &$\text{DBAR}$      & 0.032 & 0.181 & 0.398 & 0.434 && 0.007 & 0.169 & 0.440 & 0.477  \\      
                    &$\text{GPC}$        & 0.087 & \textbf{0.355} & 0.346 & 0.214 && 0.123 & \textbf{0.317} & 0.365 & 0.195  \\
{AR(1)}             &$\text{QIC}$       & 0.127 & 0.249 & 0.104 & 0.524 && 0.101 & 0.218 & 0.109 & 0.575  \\
                    &$|\text{RJ1}-1|$   & 0.099 & 0.243 & 0.398 & 0.278 && 0.048 & 0.248 & 0.386 & 0.336    \\
                    &$|\text{RJ2}-1|$   & 0.060 & 0.288 & 0.413 & 0.247 && 0.033 & 0.248 & 0.416 & 0.314     \\
                    &$\text{SC}$        & 0.237 & 0.287 & 0.250 & 0.226 && 0.269 & 0.264 & 0.271 & 0.196 \\ \hline 
                    
                    &$\text{CIC}$       & 0.039 & 0.047 & 0.117 & 0.798 && 0.011 & 0.023 & 0.066 & 0.902  \\ 
                    &$\text{DBAR}$      & 0.005 & 0.032 & 0.558 & 0.429 && 0.000 & 0.005 & 0.538 & 0.493  \\      
                    &$\text{GPC}$        & 0.092 & 0.261 & 0.493 & 0.155 && 0.145 & 0.306 & 0.416 & 0.134  \\
{Exch}              &$\text{QIC}$       & 0.167 & 0.118 & 0.238 & 0.481 && 0.114 & 0.113 & 0.231 & 0.549  \\
                    &$|\text{RJ1}-1|$   & 0.015 & 0.134 & 0.544 & 0.316 && 0.001 & 0.056 & 0.526 & 0.448   \\
                    &$|\text{RJ2}-1|$   & 0.011 & 0.073 & \textbf{0.587} & 0.331 && 0.000 & 0.014 & \textbf{0.581} & 0.414   \\
                    &$\text{SC}$        & 0.297 & 0.258 & 0.296 & 0.149 && 0.317 & 0.281 & 0.287 & 0.115\\ \hline 
                     
                    &$\text{CIC}$       & 0.043 & 0.094 & 0.039 & \textbf{0.828} && 0.018 & 0.035 & 0.026 & \textbf{0.923}  \\ 
                    &$\text{DBAR}$      & 0.012 & 0.148 & 0.441 & 0.438 && 0.001 & 0.098 & 0.487 & 0.483  \\      
                    &$\text{GPC}$        & 0.095 & 0.306 & 0.403 & 0.196 && 0.128 & 0.320 & 0.366 & 0.186  \\
{UN}                &$\text{QIC}$       & 0.164 & 0.198 & 0.133 & 0.507 && 0.129 & 0.161 & 0.123 & 0.591  \\
                    &$|\text{RJ1}-1|$   & 0.066 & 0.233 & 0.414 & 0.300 && 0.020 & 0.228 & 0.418 & 0.363  \\
                    &$|\text{RJ2}-1|$   & 0.032 & 0.206 & 0.471 & 0.297 && 0.008 & 0.189 & 0.465 & 0.353   \\
                    &$\text{SC}$        & 0.283 & 0.273 & 0.252 & 0.192 && 0.280 & 0.284 & 0.258 & 0.179 \\ \hline                      
\end{tabular}}
\end{table}

\begin{table}[ht]
\centering
\captionsetup{width=0.75\textwidth}
\caption{Proportion of selecting the true correlation structure by the competing selection criteria for unbalanced longitudinal binary data with different sample sizes when true within-subject correlation level is $\alpha=0.4$.}
\smallskip
\label{tab1M4}
\scalebox{0.7}{
\begin{tabular}{@{}llcccccccccc@{}} \hline
True&&\multicolumn{4}{c}{$N=50$} && \multicolumn{4}{c}{$N=100$}\\
\cline{3-6}\cline{8-11}
correlation &Selection&\multicolumn{4}{c}{Working correlation structure}&& \multicolumn{4}{c}{Working correlation structure}\\
\cline{3-6}\cline{8-11}
structure &criteria    &{Indep} & {AR(1)} & {Exch}  &  {UN}  && {Indep} & {AR(1)} & {Exch}  &  {UN} \\
\hline
                    &$\text{CIC}$       & 0.015 & 0.235 & 0.047 & 0.708 && 0.002 & 0.152 & 0.017 & 0.831    \\ 
                    &$\text{DBAR}$      & 0.001 & 0.262 & 0.432 & 0.332 && 0.000 & 0.247 & 0.480 & 0.324   \\      
                    &$\text{GPC}$        & 0.092 & \textbf{0.407} & 0.395 & 0.106 && 0.151 & \textbf{0.363} & 0.407 & 0.079 \\
{AR(1)}             &$\text{QIC}$       & 0.151 & 0.345 & 0.101 & 0.406 && 0.109 & 0.325 & 0.109 & 0.460 \\
                    &$|\text{RJ1}-1|$   & 0.013 & 0.270 & 0.482 & 0.245 && 0.000 & 0.291 & 0.444 & 0.284    \\
                    &$|\text{RJ2}-1|$   & 0.007 & 0.285 & 0.482 & 0.230 && 0.000 & 0.309 & 0.436 & 0.260     \\
                    &$\text{SC}$        & 0.269 & 0.342 & 0.279 & 0.110 && 0.282 & 0.330 & 0.313 & 0.075\\ \hline 
                    
                    &$\text{CIC}$       & 0.024 & 0.059 & 0.183 & 0.736 && 0.007 & 0.023 & 0.185 & 0.788    \\ 
                    &$\text{DBAR}$      & 0.000 & 0.024 & 0.674 & 0.309 && 0.000 & 0.002 & 0.663 & 0.370   \\      
                    &$\text{GPC}$        & 0.109 & 0.264 & 0.546 & 0.081 && 0.166 & 0.311 & 0.467 & 0.056 \\
{Exch}              &$\text{QIC}$       & 0.176 & 0.127 & 0.293 & 0.404 && 0.153 & 0.080 & 0.368 & 0.400  \\
                    &$|\text{RJ1}-1|$   & 0.000 & 0.093 & 0.646 & 0.270 && 0.000 & 0.013 & 0.634 & 0.365    \\
                    &$|\text{RJ2}-1|$   & 0.000 & 0.027 & \textbf{0.719} & 0.259 && 0.000 & 0.002 & \textbf{0.670} & 0.332    \\
                    &$\text{SC}$        & 0.287 & 0.280 & 0.364 & 0.070 && 0.340 & 0.264 & 0.344 & 0.052\\ \hline 
                  
                    &$\text{CIC}$       & 0.020 & 0.160 & 0.081 & \textbf{0.741} && 0.007 & 0.105 & 0.057 & \textbf{0.836}   \\ 
                    &$\text{DBAR}$      & 0.000 & 0.130 & 0.556 & 0.331 && 0.000 & 0.087 & 0.563 & 0.386   \\      
                    &$\text{GPC}$        & 0.104 & 0.342 & 0.458 & 0.096 && 0.162 & 0.333 & 0.420 & 0.086 \\
{UN}                &$\text{QIC}$       & 0.155 & 0.244 & 0.165 & 0.438 && 0.122 & 0.225 & 0.193 & 0.462 \\
                    &$|\text{RJ1}-1|$   & 0.004 & 0.245 & 0.507 & 0.258 && 0.000 & 0.171 & 0.524 & 0.324     \\
                    &$|\text{RJ2}-1|$   & 0.000 & 0.174 & 0.589 & 0.244 && 0.000 & 0.103 & 0.586 & 0.319     \\
                    &$\text{SC}$        & 0.298 & 0.310 & 0.310 & 0.082 && 0.330 & 0.292 & 0.304 & 0.074\\ \hline                      
\end{tabular}}
\end{table}

\begin{table}[ht]
\centering
\captionsetup{width=0.75\textwidth}
\caption{Proportion of selecting the true correlation structure by the competing selection criteria for balanced longitudinal count data with different sample sizes when true within-subject correlation level is $\alpha=0.2$.}
\smallskip
\label{tab1M5}
\scalebox{0.7}{
\begin{tabular}{@{}llcccccccccc@{}} \hline
True&&\multicolumn{4}{c}{$N=50$} && \multicolumn{4}{c}{$N=100$}\\
\cmidrule{3-6}\cmidrule{8-11}
correlation &Selection&\multicolumn{4}{c}{Working corr. structure}&& \multicolumn{4}{c}{Working corr. structure}\\
\cmidrule{3-6}\cmidrule{8-11}
structure &criteria    &{Indep} & {AR(1)} & {Exch}  &  {UN}  && {Indep} & {AR(1)} & {Exch}  &  {UN} \\
\midrule
                    &$\text{CIC}$       &  0.020 & 0.086 & 0.034 & 0.863 && 0.018 & 0.073 & 0.016 & 0.895\\ 
                    &$\text{DBAR}$      &  0.072 & 0.213 & 0.392 & 0.369 && 0.034 & 0.210 & 0.424 & 0.399\\      
                    &$\text{GPC}$        &  0.042 & \textbf{0.358} & 0.422 & 0.179 && 0.061 & \textbf{0.354} & 0.402 & 0.183\\
{AR(1)}             &$\text{QIC}$       &  0.137 & 0.222 & 0.085 & 0.558 && 0.118 & 0.193 & 0.080 & 0.616\\
                    &$|\text{RJ1}-1|$   &  0.226 & 0.263 & 0.320 & 0.201 && 0.119 & 0.311 & 0.331 & 0.251\\
                    &$|\text{RJ2}-1|$   &  0.159 & 0.306 & 0.337 & 0.204 && 0.085 & 0.309 & 0.357 & 0.261\\
                    &$\text{SC}$        &  0.348 & 0.281 & 0.187 & 0.186 && 0.386 & 0.276 & 0.153 & 0.185\\ \hline 
                    
                    &$\text{CIC}$       &  0.022 & 0.050 & 0.059 & 0.871 && 0.008 & 0.018 & 0.036 & 0.941  \\ 
                    &$\text{DBAR}$      &  0.006 & 0.065 & 0.550 & 0.409 && 0.000 & 0.014 & \textbf{0.561} & 0.487\\      
                    &$\text{GPC}$        &  0.087 & 0.136 & \textbf{0.610} & 0.167 && 0.155 & 0.143 & 0.543 & 0.159\\
{Exch}              &$\text{QIC}$       &  0.138 & 0.100 & 0.158 & 0.607 && 0.117 & 0.076 & 0.151 & 0.659\\
                    &$|\text{RJ1}-1|$   &  0.048 & 0.217 & 0.481 & 0.265 && 0.003 & 0.124 & 0.537 & 0.361   \\
                    &$|\text{RJ2}-1|$   &  0.025 & 0.153 & 0.556 & 0.268 && 0.002 & 0.070 & 0.556 & 0.384   \\
                    &$\text{SC}$        &  0.347 & 0.209 & 0.290 & 0.155 && 0.400 & 0.176 & 0.289 & 0.136\\ \hline 
                    
                    &$\text{CIC}$       &  0.033 & 0.073 & 0.039 & \textbf{0.856} && 0.017 & 0.056 & 0.041 & \textbf{0.887}   \\ 
                    &$\text{DBAR}$      &  0.028 & 0.180 & 0.435 & 0.394 && 0.012 & 0.147 & 0.449 & 0.46  \\      
                    &$\text{GPC}$        &  0.044 & 0.247 & 0.524 & 0.185 && 0.108 & 0.235 & 0.437 & 0.220\\
{UN}                &$\text{QIC}$       &  0.140 & 0.176 & 0.127 & 0.558 && 0.113 & 0.158 & 0.134 & 0.599\\
                    &$|\text{RJ1}-1|$   &  0.122 & 0.267 & 0.376 & 0.246 && 0.060 & 0.319 & 0.366 & 0.280   \\
                    &$|\text{RJ2}-1|$   &  0.088 & 0.241 & 0.425 & 0.252 && 0.039 & 0.277 & 0.390 & 0.303   \\
                    &$\text{SC}$        &  0.367 & 0.268 & 0.199 & 0.167 && 0.378 & 0.253 & 0.169 & 0.200\\ \hline                      
\end{tabular}}
\end{table}

\begin{table}[ht]
\centering
\captionsetup{width=0.75\textwidth}
\caption{Proportion of selecting the true correlation structure by the competing selection criteria for balanced longitudinal count data with different sample sizes when true within-subject correlation level is $\alpha=0.4$.}
\smallskip
\label{tab1M6}
\scalebox{0.7}{
\begin{tabular}{@{}llcccccccccc@{}} \hline
True&&\multicolumn{4}{c}{$N=50$} && \multicolumn{4}{c}{$N=100$}\\
\cline{3-6}\cline{8-11}
correlation &Selection&\multicolumn{4}{c}{Working corr. structure}&& \multicolumn{4}{c}{Working corr. structure}\\
\cline{3-6}\cline{8-11}
structure &criteria    &{Indep} & {AR(1)} & {Exch}  &  {UN}  && {Indep} & {AR(1)} & {Exch}  &  {UN} \\
\hline
                    &$\text{CIC}$       & 0.010 & 0.194 & 0.026 & 0.773 && 0.003 & 0.105 & 0.017 & 0.878   \\ 
                    &$\text{DBAR}$      & 0.008 & 0.241 & 0.481 & 0.288 && 0.000 & 0.220 & 0.499 & 0.346   \\      
                    &$\text{GPC}$        & 0.106 & \textbf{0.452} & 0.364 & 0.078 && 0.188 & \textbf{0.428} & 0.315 & 0.069 \\
{AR(1)}             &$\text{QIC}$       & 0.098 & 0.288 & 0.133 & 0.481 && 0.073 & 0.297 & 0.097 & 0.534 \\
                    &$|\text{RJ1}-1|$   & 0.059 & 0.313 & 0.433 & 0.205 && 0.006 & 0.337 & 0.441 & 0.229    \\
                    &$|\text{RJ2}-1|$   & 0.034 & 0.310 & 0.447 & 0.214 && 0.003 & 0.347 & 0.447 & 0.212     \\
                    &$\text{SC}$        & 0.336 & 0.424 & 0.161 & 0.079 && 0.379 & 0.413 & 0.148 & 0.060\\ \hline 
                    
                    &$\text{CIC}$       & 0.011 & 0.063 & 0.074 & 0.853 && 0.004 & 0.020 & 0.041 & 0.937   \\ 
                    &$\text{DBAR}$      & 0.000 & 0.049 & 0.617 & 0.344 && 0.000 & 0.011 & 0.598 & 0.443   \\      
                    &$\text{GPC}$        & 0.175 & 0.140 & 0.623 & 0.063 && 0.271 & 0.155 & 0.527 & 0.047 \\
{Exch}              &$\text{QIC}$       & 0.152 & 0.075 & 0.211 & 0.567 && 0.122 & 0.054 & 0.194 & 0.631\\
                    &$|\text{RJ1}-1|$   & 0.004 & 0.201 & 0.589 & 0.214 && 0.000 & 0.082 & 0.625 & 0.319   \\
                    &$|\text{RJ2}-1|$   & 0.000 & 0.125 & \textbf{0.669} & 0.211 && 0.000 & 0.036 & \textbf{0.680} & 0.288    \\
                    &$\text{SC}$        & 0.335 & 0.177 & 0.429 & 0.059 && 0.386 & 0.169 & 0.401 & 0.044\\ \hline 
                     
                    &$\text{CIC}$       & 0.009 & 0.107 & 0.057 & \textbf{0.830} && 0.002 & 0.047 & 0.022 & \textbf{0.929}   \\ 
                    &$\text{DBAR}$      & 0.002 & 0.169 & 0.504 & 0.347 && 0.000 & 0.079 & 0.570 & 0.394  \\      
                    &$\text{GPC}$        & 0.141 & 0.300 & 0.460 & 0.100 && 0.242 & 0.322 & 0.365 & 0.071 \\
{UN}                &$\text{QIC}$       & 0.116 & 0.160 & 0.164 & 0.562 && 0.113 & 0.151 & 0.147 & 0.590 \\
                    &$|\text{RJ1}-1|$   & 0.023 & 0.378 & 0.423 & 0.182 && 0.002 & 0.255 & 0.470 & 0.291    \\
                    &$|\text{RJ2}-1|$   & 0.009 & 0.300 & 0.500 & 0.197 && 0.000 & 0.189 & 0.529 & 0.286   \\
                    &$\text{SC}$        & 0.355 & 0.357 & 0.215 & 0.074 && 0.387 & 0.342 & 0.211 & 0.060\\ \hline                      
\end{tabular}}
\end{table}

\begin{table}[ht]
\centering
\captionsetup{width=0.75\textwidth}
\caption{Proportion of selecting the true correlation structure by the competing selection criteria for unbalanced longitudinal count data with different sample sizes when true within-subject correlation level is $\alpha=0.2$.}
\smallskip
\label{tab1M7}
\scalebox{0.7}{
\begin{tabular}{@{}llcccccccccc@{}} \hline
True&&\multicolumn{4}{c}{$N=50$} && \multicolumn{4}{c}{$N=100$}\\
\cmidrule{3-6}\cmidrule{8-11}
correlation &Selection&\multicolumn{4}{c}{Working corr. structure}&& \multicolumn{4}{c}{Working corr. structure}\\
\cmidrule{3-6}\cmidrule{8-11}
structure &criteria    &{Indep} & {AR(1)} & {Exch}  &  {UN}  && {Indep} & {AR(1)} & {Exch}  &  {UN} \\
\midrule
                    &$\text{CIC}$       & 0.041 & 0.109 & 0.031 & 0.823 && 0.019 & 0.075 & 0.022 & 0.888   \\ 
                    &$\text{DBAR}$      & 0.065 & 0.235 & 0.388 & 0.366 && 0.033 & 0.228 & 0.418 & 0.402   \\      
                    &$\text{GPC}$        & 0.041 & \textbf{0.393} & 0.410 & 0.156 && 0.089 & \textbf{0.362} & 0.388 & 0.162 \\
{AR(1)}             &$\text{QIC}$       & 0.135 & 0.197 & 0.118 & 0.553 && 0.117 & 0.196 & 0.089 & 0.604  \\
                    &$|\text{RJ1}-1|$   & 0.203 & 0.269 & 0.318 & 0.226 && 0.134 & 0.299 & 0.343 & 0.238    \\
                    &$|\text{RJ2}-1|$   & 0.136 & 0.303 & 0.322 & 0.241 && 0.092 & 0.308 & 0.353 & 0.256    \\
                    &$\text{SC}$        & 0.306 & 0.311 & 0.229 & 0.154 && 0.314 & 0.301 & 0.221 & 0.164\\ \hline 
                    
                    &$\text{CIC}$       & 0.031 & 0.064 & 0.067 & 0.841 && 0.015 & 0.026 & 0.046 & 0.916   \\ 
                    &$\text{DBAR}$      & 0.010 & 0.077 & 0.532 & 0.407 && 0.000 & 0.039 & \textbf{0.553} & 0.470  \\      
                    &$\text{GPC}$        & 0.070 & 0.138 & \textbf{0.641} & 0.152 && 0.142 & 0.172 & 0.547 & 0.139\\
{Exch}              &$\text{QIC}$       & 0.157 & 0.112 & 0.173 & 0.561 && 0.130 & 0.068 & 0.161 & 0.648 \\
                    &$|\text{RJ1}-1|$   & 0.066 & 0.274 & 0.432 & 0.241 && 0.013 & 0.175 & 0.493 & 0.342    \\
                    &$|\text{RJ2}-1|$   & 0.032 & 0.191 & 0.519 & 0.264 && 0.006 & 0.120 & 0.539 & 0.346     \\
                    &$\text{SC}$        & 0.331 & 0.183 & 0.347 & 0.140 && 0.366 & 0.190 & 0.320 & 0.124\\ \hline 
                     
                    &$\text{CIC}$       & 0.034 & 0.089 & 0.065 & \textbf{0.813} && 0.017 & 0.063 & 0.036 & \textbf{0.885}    \\ 
                    &$\text{DBAR}$      & 0.041 & 0.193 & 0.468 & 0.351 && 0.017 & 0.153 & 0.443 & 0.451  \\      
                    &$\text{GPC}$        & 0.051 & 0.314 & 0.485 & 0.150 && 0.090 & 0.279 & 0.450 & 0.181 \\
{UN}                &$\text{QIC}$       & 0.155 & 0.157 & 0.152 & 0.539 && 0.118 & 0.162 & 0.112 & 0.611 \\
                    &$|\text{RJ1}-1|$   & 0.152 & 0.277 & 0.367 & 0.213 && 0.075 & 0.329 & 0.364 & 0.256   \\
                    &$|\text{RJ2}-1|$   & 0.087 & 0.284 & 0.406 & 0.233 && 0.048 & 0.293 & 0.384 & 0.286   \\
                    &$\text{SC}$        & 0.323 & 0.297 & 0.240 & 0.142 && 0.332 & 0.269 & 0.243 & 0.157\\ \hline                      
\end{tabular}}
\end{table}

\begin{table}[ht]
\centering
\captionsetup{width=0.75\textwidth}
\caption{Proportion of selecting the true correlation structure by the competing selection criteria for unbalanced longitudinal count data with different sample sizes when true within-subject correlation level is $\alpha=0.4$.}
\smallskip
\label{tab1M8}
\scalebox{0.7}{
\begin{tabular}{@{}llcccccccccc@{}} \hline
True&&\multicolumn{4}{c}{$N=50$} && \multicolumn{4}{c}{$N=100$}\\
\cline{3-6}\cline{8-11}
correlation &Selection&\multicolumn{4}{c}{Working corr. structure}&& \multicolumn{4}{c}{Working corr. structure}\\
\cline{3-6}\cline{8-11}
structure &criteria    &{Indep} & {AR(1)} & {Exch}  &  {UN}  && {Indep} & {AR(1)} & {Exch}  &  {UN} \\
\hline
                    &$\text{CIC}$       & 0.015 & 0.206 & 0.050 & 0.735 && 0.002 & 0.149 & 0.013 & 0.841   \\ 
                    &$\text{DBAR}$      & 0.013 & 0.291 & 0.466 & 0.257 && 0.001 & 0.262 & 0.484 & 0.311   \\      
                    &$\text{GPC}$        & 0.100 & \textbf{0.452} & 0.378 & 0.071 && 0.195 & \textbf{0.422} & 0.326 & 0.057 \\
{AR(1)}             &$\text{QIC}$       & 0.141 & 0.280 & 0.134 & 0.448 && 0.080 & 0.300 & 0.122 & 0.501  \\
                    &$|\text{RJ1}-1|$   & 0.069 & 0.326 & 0.428 & 0.190 && 0.014 & 0.345 & 0.428 & 0.232    \\
                    &$|\text{RJ2}-1|$   & 0.034 & 0.375 & 0.424 & 0.174 && 0.005 & 0.345 & 0.440 & 0.219   \\
                    &$\text{SC}$        & 0.295 & 0.433 & 0.205 & 0.068 && 0.356 & 0.391 & 0.197 & 0.056  \\ \hline 
                    
                    &$\text{CIC}$       & 0.015 & 0.072 & 0.132 & 0.784 && 0.004 & 0.029 & 0.079 & 0.891   \\ 
                    &$\text{DBAR}$      & 0.000 & 0.066 & 0.601 & 0.357 && 0.000 & 0.020 & 0.629 & 0.400   \\      
                    &$\text{GPC}$        & 0.126 & 0.142 & \textbf{0.672} & 0.062 && 0.259 & 0.175 & 0.537 & 0.030\\
{Exch}              &$\text{QIC}$       & 0.164 & 0.085 & 0.254 & 0.500 && 0.117 & 0.058 & 0.245 & 0.584  \\
                    &$|\text{RJ1}-1|$   & 0.010 & 0.300 & 0.490 & 0.212 && 0.000 & 0.128 & 0.601 & 0.286    \\
                    &$|\text{RJ2}-1|$   & 0.003 & 0.189 & 0.606 & 0.208 && 0.000 & 0.076 & \textbf{0.655} & 0.276    \\
                    &$\text{SC}$        & 0.264 & 0.193 & 0.487 & 0.058 && 0.368 & 0.189 & 0.418 & 0.026\\ \hline 
                     
                    &$\text{CIC}$       & 0.013 & 0.146 & 0.073 & \textbf{0.768} && 0.003 & 0.071 & 0.040 & \textbf{0.887}   \\ 
                    &$\text{DBAR}$      & 0.003 & 0.210 & 0.514 & 0.299 && 0.000 & 0.103 & 0.558 & 0.381   \\      
                    &$\text{GPC}$        & 0.124 & 0.327 & 0.474 & 0.076 && 0.223 & 0.295 & 0.426 & 0.056  \\
{UN}                &$\text{QIC}$       & 0.142 & 0.180 & 0.196 & 0.485 && 0.115 & 0.161 & 0.150 & 0.577 \\
                    &$|\text{RJ1}-1|$   & 0.022 & 0.400 & 0.421 & 0.167 && 0.005 & 0.314 & 0.452 & 0.242    \\
                    &$|\text{RJ2}-1|$   & 0.011 & 0.316 & 0.492 & 0.184 && 0.001 & 0.238 & 0.533 & 0.233     \\
                    &$\text{SC}$        & 0.317 & 0.344 & 0.276 & 0.063 && 0.351 & 0.314 & 0.289 & 0.046 \\ \hline                      
\end{tabular}}
\end{table}

\begin{table}[ht]
\captionsetup{width=0.75\textwidth}
\caption{Proportion of selecting the true correlation structure among the candidate working correlation structures {Indep}, {AR(1)}, {Exch} by the competing selection criteria for balanced longitudinal binary data with different sample sizes when true within-subject correlation level is $\alpha=0.2$.}
\smallskip
\centering
\label{tab1M9}
\scalebox{0.8}{
\begin{tabular}{llccccccc} \hline
True&&\multicolumn{3}{c}{$N=50$} && \multicolumn{3}{c}{$N=100$}\\
\cmidrule{3-5}\cmidrule{7-9}
correlation &Selection&\multicolumn{3}{c}{Working corr. structure}&& \multicolumn{3}{c}{Working corr. structure}\\ \cmidrule{3-5}\cmidrule{7-9}
structure &criteria             &{Indep}& {AR(1)} & {Exch}  && {Indep} & {AR(1)} & {Exch}  \\
\hline
            &$\text{CIC}$       & 0.116 & \textbf{0.705} & 0.196 && 0.078  & \textbf{0.790} & 0.139\\
            &$\text{DBAR}$      & 0.057 & 0.360 & 0.648 && 0.016  & 0.359  & 0.695 \\
            &$\text{GPC}$        & 0.049 & 0.487 & 0.464 && 0.042  & 0.467  & 0.491  \\
{AR(1)}     &$\text{QIC}$       & 0.221 & 0.564 & 0.226 && 0.188  & 0.598  & 0.221  \\
            &$|\text{RJ1}-1|$   & 0.110 & 0.303 & 0.599 && 0.064  & 0.324  & 0.620\\
            &$|\text{RJ2}-1|$   & 0.068 & 0.331 & 0.605 && 0.033  & 0.340  & 0.632\\
            &$\text{SC}$        & 0.300 & 0.453 & 0.250 && 0.310  & 0.437  & 0.255\\ \hline
                    
          &$\text{CIC}$         & 0.151 & 0.246 & 0.607 && 0.071 & 0.190 & 0.743\\
          &$\text{DBAR}$        & 0.003 & 0.054 & \textbf{0.951} && 0.000 & 0.006 & \textbf{0.995} \\    
          &$\text{GPC}$          & 0.035 & 0.274 & 0.691 && 0.049 & 0.349 & 0.602\\
{Exch}    &$\text{QIC}$         & 0.290 & 0.216 & 0.500 && 0.216 & 0.172 & 0.615  \\
          &$|\text{RJ1}-1|$     & 0.017 & 0.124 & 0.861 && 0.000 & 0.041 & 0.960\\
          &$|\text{RJ2}-1|$     & 0.009 & 0.071 & 0.920 && 0.000 & 0.012 & 0.989 \\ 
          &$\text{SC}$          & 0.363 & 0.391 & 0.247 && 0.303 & 0.418 & 0.279  \\ \hline         
\end{tabular}}
\end{table}

\begin{table}[ht]
\centering
\captionsetup{width=0.75\textwidth}
\caption{Values of selection criteria ({CIC}, {DBAR}, {GPC}, {QIC}, {RJ1}, {RJ2}, and {SC}) under different working correlation structures (Indep, AR(1), Exch, and UN) for the analysis of the CARDIA study.}
\smallskip  
\label{tab:tab2M3} 
\scalebox{0.8}{
\begin{tabular} {@{}llccccc@{}} \hline
         &
Selection&&\multicolumn{4}{c}{Working correlation structure}\\ \cline{3-7}
Race/Gender       &criteria && {Indep}  & {AR(1)} & {Exch} &  {UN}  \\ 
\hline 

                    &$\text{CIC}$       && 22.259    & 20.580            &  20.699          & \textbf{20.408}  \\ 
                    &$\text{DBAR}$      && 6.582     & 0.083             &  0.033           & \textbf{0.006}   \\      
                    &$\text{GPC}$        && 7217.867  & \textbf{6834.430} &  6900.614        & 7100.406 \\
Black females       &$\text{QIC}$       && 8228.186  & \textbf{8226.052} &  8226.733        & 8226.345 \\
                    &$|\text{RJ1}-1|$   && 1.782     & 0.045             &  0.056           & \textbf{0.007}    \\
                    &$|\text{RJ2}-1|$   && 10.147    & 0.174             &  0.144           & \textbf{0.008}    \\
                    &$\text{SC}$        && 7173.059  & \textbf{6817.624} &  6883.653        & 7084.452  \\ \hline 

                    &$\text{CIC}$       && 21.576    & 19.703            & 19.705           & \textbf{19.343}\\ 
                    &$\text{DBAR}$      && 5.936     & 0.056             & 0.095            & \textbf{0.003} \\      
                    &$\text{GPC}$        && 5367.099  & \textbf{5101.026} & 5214.248         & 5337.206\\
{Black males}       &$\text{QIC}$       && 6532.317  & 6531.367          & \textbf{6530.689}& 6530.804\\
                    &$|\text{RJ1}-1|$   && 1.697     & 0.027             & 0.069            & \textbf{0.000}\\
                    &$|\text{RJ2}-1|$   && 9.330     & 0.109             & 0.234            & \textbf{0.004} \\
                    &$\text{SC}$        && 5323.583  & \textbf{5084.500} & 5197.047         & 5321.113 \\ \hline    

                    &$\text{CIC}$       &&  21.233   & 20.147            & 20.235            & \textbf{19.861}  \\ 
                    &$\text{DBAR}$      &&  5.632    & 0.129             & 0.037             & \textbf{0.012}  \\      
                    &$\text{GPC}$        &&  6918.060 & 6620.616          & \textbf{6392.052} & 6655.366  \\
White females       &$\text{QIC}$       &&  6498.943 & \textbf{6498.869} & 6500.929          & 6499.063  \\
                    &$|\text{RJ1}-1|$   &&  1.654    & 0.098             & 0.052             & \textbf{0.014}  \\
                    &$|\text{RJ2}-1|$   &&  8.941    & 0.325             & 0.141             & \textbf{0.016}  \\
                    &$\text{SC}$        &&  6875.203 & 6602.925          & \textbf{6375.116} & 6639.486 \\
\hline 
                    &$\text{CIC}$       &&  22.169              & 20.973            & 21.048    & \textbf{20.788}  \\ 
                    &$\text{DBAR}$      &&  6.448               & 0.103             & 0.072     & \textbf{0.017}   \\      
                    &$\text{GPC}$        &&  6157.361            & \textbf{5321.125} & 5411.959  & 5592.362  \\
White males         &$\text{QIC}$       &&  \textbf{5949.163}   & 5950.807          & 5950.156  & 5949.231  \\
                    &$|\text{RJ1}-1|$   &&  1.771               & 0.045             & 0.058     & \textbf{0.016}  \\
                    &$|\text{RJ2}-1|$   &&  9.990               & 0.193             & 0.189     & \textbf{0.015}  \\
                    &$\text{SC}$        &&  6112.547            & \textbf{5304.263} & 5394.906  & 5576.505  \\
\hline 
\multicolumn{7}{l}{Note that: The smaller value of each criterion leads to a better working}\\
\multicolumn{7}{l}{correlation structure.}
\end{tabular}}
\end{table}

\begin{table}[ht]
\caption{GEE estimates (emp.\ se's) for the CARDIA study.}
\smallskip
\label{tab3}
\begin{center}
\scalebox{0.8}{
\begin{tabular}{@{}lccccccc@{}} \hline
                        &            &&   \multicolumn{2}{c}{Black}       && \multicolumn{2}{c}{White} \\ \cline{4-5}\cline{7-8}
{Variable}              &{Parameter} && {Females}       & {Males}         && {Females}       & {Males} \\ \hline
Intercept               &$\beta_{0}$ && -0.011  (0.111) &  0.164  (0.110) &&  0.164  (0.147) &  0.295  (0.155) \\
Age/10                  &$\beta_{1}$ &&  0.327* (0.157) &  0.648* (0.176) && -0.244  (0.186) & -0.079  (0.193) \\
Age-square              &$\beta_{2}$ && -0.593  (0.430) &  0.319  (0.467) && -0.305  (0.512) & -0.692  (0.573) \\
Some college$^{a}$   &$\beta_{3}$ && -0.693* (0.117) & -0.964* (0.125) && -0.590* (0.161) & -0.987* (0.169) \\
College degree$^{a}$ &$\beta_{4}$ && -1.769* (0.170) & -1.787* (0.177) && -1.871* (0.161) & -2.005* (0.167) \\
Year (Yr/10)            &$\beta_{L}$ &&  0.468* (0.220) &  0.449  (0.248) && -0.713* (0.266) &  0.022  (0.248)  \\
Year-square             &$\beta_{Q}$ && -0.935* (0.388) & -0.451  (0.429) &&  0.516  (0.475) & -0.290  (0.437) \\
Year-cubic              &$\beta_{C}$ &&  0.353* (0.174) &  0.089  (0.191) && -0.203  (0.215) &  0.068  (0.195) \\ \hline
Within-subject          & \multirow{2}{*}{$\alpha$}   &&   \multirow{2}{*}{0.779}    &    \multirow{2}{*}{0.790}          &&                 \multirow{2}{*}{0.617} & \multirow{2}{*}{0.737}              \\ 
association             &            &&                 &                 &&                 &              \\ \hline  
\multicolumn{8}{l}{Note that: Ar(1) structure is used for black females, black males, and white males}\\ 
\multicolumn{8}{l}{and Exch structure is for white females.}\\
\multicolumn{8}{l}{$*$ P-value $<$ $0.05$.}\\
\multicolumn{8}{l}{$^{a}$No college attended is the reference category.}\\
\end{tabular}}
\end{center}
\end{table}

\begin{table}[ht]
\centering
\caption{Test for any trend and Yr 0 vs Yr 1.5.}
\smallskip
\label{tab4}
\begin{tabular}{@{}llcccccc@{}}
\toprule
             &&\multicolumn{2}{c}{Black} && \multicolumn{2}{c}{White} \\
\cmidrule{3-4}\cmidrule{6-7}
             & & Females & Males&& Females &Males \\
\midrule
Score test:  & $\chi^2_3$  & 28.49     & 13.66      && 63.86    & 37.99 \\
\,\, Trend   & P-value     & $<0.001$  & $<0.050$   && $<0.001$ & $<0.001$\\
\midrule
Score test:  & $\chi^2_1$  & 13.06     & 0.31       && 58.43    & 34.23  \\
\,\, Yr 0 vs Yr 1.5 &P-value& $<0.001$  &$0.580$     && $<0.001$ & $<0.001$ \\ \hline
\end{tabular}
\end{table}

\clearpage
\newpage
 \vspace{2.5cm}
    \begin{figure}[h]
       \centering
        \includegraphics[width=0.75\textwidth]{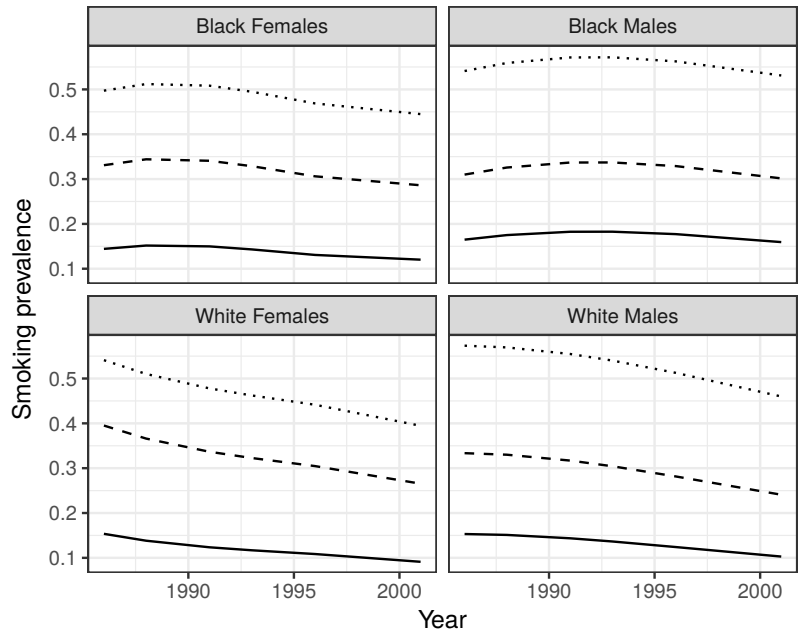}
        \caption{Model predicted trends in smoking among a cohort of young adults, the CARDIA study; solid line indicates those who attained a college degree, long dashes indicates some college, and short dashes indicates high school education or less.}\label{fig:fig1}
    \end{figure}

\clearpage    
\newpage
\appendix
\numberwithin{remark}{section}
\section{Details of the GPC statistic}

\begin{theorem}
Define $\hat{\mathbf{r}}_{(i)}=\textbf{Z}_i-\textbf{X}_i\hat{\bm{\beta}}_{(i)}$ as the linearized {PRESS} residual vector for the $i$th cluster. 
The $\text{GPC}$ statistic approximates
$\sum_i^{N} \hat{\mathbf{r}}_{(i)}'\{\hat{\text{Cov}}(\hat{\mathbf{r}}_{(i)})\}^{-1}\hat{\mathbf{r}}_{(i)}$ and
$\sum_i^{N} \hat{\mathbf{e}}_{(i)}' \{\hat{\text{Cov}}(\hat{\mathbf{e}}_{(i)})\}^{-1} \hat{\mathbf{e}}_{(i)},$
where $\hat{\mathbf{e}}_{(i)}=\mathbf{y}_i-g^{-1}(\textbf{X}_i\hat{\bm{\beta}}_{(i)}).$
\end{theorem}

\begin{proof}
The relation $\hat{\mathbf{r}}_{(i)}=\textbf{X}_i(\hat{\bm{\beta}}-\hat{\bm{\beta}}_{(i)})+\textbf{L}_i\hat{\mathbf{e}}_i$ follows
by substitution of $\textbf{Z}_i=\textbf{X}_i\hat{\bm{\beta}}+\textbf{L}_i\hat{\mathbf{e}}_i.$  Next, define
$\textbf{C}_i=\textbf{M}^{-1}\textbf{D}_i'\textbf{V}_i^{-1}(\textbf{I}-\textbf{H}_i)^{-1}\hat{\mathbf{e}}_i$,
which is a one-step approximation of $\hat{\bm{\beta}}-\hat{\bm{\beta}}_{(i)}$ \citep{preisser2008note} and reduces
to DBETA for multiple linear regression \citep{belsey1980regression}.
Note that $\textbf{D}_i = \textbf{L}_i^{-1}\textbf{X}_i$.  By substitution,
$\hat{\mathbf{r}}_{(i)}  \approx \textbf{L}_i\big[\textbf{L}_i^{-1}\textbf{X}_i\textbf{M}^{-1}\textbf{D}_i'\textbf{V}_i^{-1}(\textbf{I}-\textbf{H}_i)^{-1}+\textbf{I}_{n_i}\big]\hat{\mathbf{e}}_i
=\textbf{L}_i\big[\textbf{H}_i(\textbf{I}-\textbf{H}_i)^{-1}+\textbf{I}\big]\hat{\mathbf{e}}_i
=\textbf{L}_i(\textbf{I}-\textbf{H}_i)^{-1}\hat{\mathbf{e}}_i.$ It follows that
$\text{Cov}(\hat{\mathbf{r}}_{(i)})\approx
\textbf{L}_i (\textbf{I}-\textbf{H}_i)^{-1}\text{Cov}(\hat{\mathbf{e}}_i) (\textbf{I}-\textbf{H}'_i)^{-1} \textbf{L}_i \approx \textbf{L}_i\textbf{V}_i\textbf{L}_i.$
Thus $\sum_i^{N}\hat{\mathbf{r}}_{(i)}'\{\hat{\text{Cov}}(\hat{\mathbf{r}}_{(i)})\}^{-1}\hat{\mathbf{r}}_{(i)} \approx \text{GPC}.$

Next, a first-order linear Taylor series expansion of $g^{-1}(\textbf{X}_i\hat{\bm{\beta}}_{(i)})$ about $\hat{\bm{\beta}}$ is
$$g^{-1}(\textbf{X}_i\hat{\bm{\beta}}_{(i)}) \approx g^{-1}(\textbf{X}_i\hat{\bm{\beta}}) - \textbf{L}_i^{-1}\textbf{X}_i(\hat{\bm{\beta}}-\hat{\bm{\beta}}_{(i)}).$$
Applying $\textbf{C}_i$ and algebra yields $\hat{\mathbf{e}}_{(i)} \approx \hat{\mathbf{e}}_{i} + \textbf{H}_i(\textbf{I}-\textbf{H}_i)^{-1}\hat{\mathbf{e}}_{i}
=(\textbf{I}-\textbf{H}_i)^{-1}\hat{\mathbf{e}}_{i}.$  Thus,
$\hat{\text{Cov}}(\hat{\mathbf{e}}_{(i)}) \approx (\textbf{I}-\textbf{H}_i)^{-1} \hat{\text{Cov}}(\hat{\mathbf{e}}_i) (\textbf{I}-\textbf{H}'_i)^{-1} \approx \hat{\textbf{V}}_i$
and $\sum_i^{N} \hat{\mathbf{e}}_{(i)}' \{\hat{\text{Cov}}(\hat{\mathbf{e}}_{(i)})\}^{-1} \hat{\mathbf{e}}_{(i)} \linebreak \approx \text{GPC}.$\\
\end{proof}

\newpage
\section{Mean squared error values}
\setcounter{table}{0}
\renewcommand{\thetable}{B\arabic{table}}

Appendix provides the results for mean squared error values of regression parameters for each $48$ different scenario above via Tables~\ref{tab1S1}-\ref{tab1S8} given below.

\begin{table}[ht]
\centering
\caption{Mean squared error values of parameters for the case of balanced longitudinal binary data with different sample sizes when true within-subject correlation level is $\alpha=0.2$.}
\smallskip
\label{tab1S1}
\scalebox{0.7}{
\begin{tabular}{@{}llcccccccccc@{}} \hline
True&&\multicolumn{4}{c}{$N=50$} && \multicolumn{4}{c}{$N=100$}\\
\cline{3-6}\cline{8-11}
correlation &Param.&\multicolumn{4}{c}{Working corr.}&& \multicolumn{4}{c}{Working corr. structure}\\
\cline{3-6}\cline{8-11}
structure &   &{Indep}     &{AR(1)} & {Exch} & {UN}  && {Indep}     &{AR(1)} & {Exch} & {UN} \\
\hline
                    &$\beta_{0}$        &  0.084 & 0.083 & 0.084 & 0.093 && 0.040 & 0.041 & 0.042 & 0.042  \\
{AR(1)}             &$\beta_{1}$        &  0.150 & 0.148 & 0.150 & 0.154 && 0.070 & 0.067 & 0.068 & 0.070 \\
                    &$\beta_{2}$        &  0.103 & 0.100 & 0.102 & 0.116 && 0.050 & 0.047 & 0.049 & 0.049  \\ \hline 
                    
                    &$\beta_{0}$        &  0.113 & 0.113 & 0.110 & 0.124 && 0.054 & 0.053 & 0.052 & 0.056\\
{Exch}              &$\beta_{1}$        &  0.199 & 0.199 & 0.198 & 0.214 && 0.090 & 0.091 & 0.090 & 0.097\\
                    &$\beta_{2}$        &  0.107 & 0.105 & 0.099 & 0.112 && 0.049 & 0.047 & 0.045 & 0.048  \\ \hline 
                     
                    &$\beta_{0}$        &  0.094 & 0.094 & 0.094 & 0.105 && 0.049 & 0.048 & 0.048 & 0.050\\ 
{UN}                &$\beta_{1}$        &  0.169 & 0.167 & 0.169 & 0.182 && 0.083 & 0.083 & 0.084 & 0.086\\      
                    &$\beta_{2}$        &  0.103 & 0.098 & 0.099 & 0.112 && 0.053 & 0.050 & 0.050 & 0.052\\   \hline 
                                
\end{tabular}}
\end{table}

\begin{table}[ht]
\centering
\caption{Mean squared error values of parameters for the case of balanced longitudinal binary data with different sample sizes when true within-subject correlation level is $\alpha=0.4$.}
\smallskip
\label{tab1S2}
\scalebox{0.7}{
\begin{tabular}{@{}llcccccccccc@{}} \hline
True&&\multicolumn{4}{c}{$N=50$} && \multicolumn{4}{c}{$N=100$}\\
\cline{3-6}\cline{8-11}
correlation &Param.&\multicolumn{4}{c}{Working corr. structure}&& \multicolumn{4}{c}{Working corr. structure}\\
\cline{3-6}\cline{8-11}
structure &    &{Indep}     &{AR(1)} & {Exch} & {UN}  && {Indep} &{AR(1)} & {Exch} & {UN} \\
\hline
                    &$\beta_{0}$        &  0.109 & 0.098 & 0.104 & 0.109 && 0.050 & 0.048 & 0.051 & 0.052    \\ 
{AR(1)}             &$\beta_{1}$        &  0.204 & 0.198 & 0.202 & 0.215 && 0.100 & 0.096 & 0.100 & 0.102    \\      
                    &$\beta_{2}$        &  0.116 & 0.089 & 0.105 & 0.101 && 0.050 & 0.038 & 0.043 & 0.041    \\  \hline  
                    
                    &$\beta_{0}$        &  0.157 & 0.156 & 0.149 & 0.164 && 0.073 & 0.071 & 0.070 & 0.079  \\ 
{Exch}              &$\beta_{1}$        &  0.296 & 0.300 & 0.293 & 0.317 && 0.134 & 0.135 & 0.132 & 0.14  \\      
                    &$\beta_{2}$        &  0.108 & 0.090 & 0.079 & 0.084 && 0.057 & 0.045 & 0.040 & 0.042   \\ \hline   
                    
                    &$\beta_{0}$        &  0.139 & 0.134 & 0.135 & 0.155 && 0.058 & 0.056 & 0.057 & 0.059\\ 
{UN}                &$\beta_{1}$        &  0.241 & 0.238 & 0.240 & 0.265 && 0.117 & 0.115 & 0.117 & 0.124\\      
                    &$\beta_{2}$        &  0.115 & 0.090 & 0.095 & 0.104 && 0.051 & 0.041 & 0.042 & 0.043\\  \hline 
                                                                   
\end{tabular}}
\end{table}

\begin{table}[ht]
\centering
\caption{Mean squared error values of parameters for the case of unbalanced longitudinal binary data with different sample sizes when true within-subject correlation level is $\alpha=0.2$.}
\smallskip
\label{tab1S3}
\scalebox{0.7}{
\begin{tabular}{@{}llcccccccccc@{}} \hline
True&&\multicolumn{4}{c}{$N=50$} && \multicolumn{4}{c}{$N=100$}\\
\cline{3-6}\cline{8-11}
correlation &Param.&\multicolumn{4}{c}{Working corr.}&& \multicolumn{4}{c}{Working corr. structure}\\
\cline{3-6}\cline{8-11}
structure &&{Indep}     &{AR(1)} & {Exch} & {UN}  && {Indep}     &{AR(1)} & {Exch} & {UN} \\
\hline
                    &$\beta_{0}$        &  0.101 & 0.094 & 0.100 & 0.113 && 0.050 & 0.046 & 0.047 & 0.049  \\ 
{AR(1)}             &$\beta_{1}$        &  0.159 & 0.164 & 0.160 & 0.171 && 0.080 & 0.080 & 0.081 & 0.085  \\      
                    &$\beta_{2}$        &  0.115 & 0.120 & 0.114 & 0.128 && 0.050 & 0.052 & 0.052 & 0.055  \\ \hline   
                    
                    &$\beta_{0}$        &  0.111 & 0.111 & 0.108 & 0.119 && 0.054 & 0.054 & 0.052 & 0.056  \\ 
{Exch}              &$\beta_{1}$        &  0.205 & 0.206 & 0.203 & 0.209 && 0.094 & 0.095 & 0.095 & 0.101  \\      
                    &$\beta_{2}$        &  0.129 & 0.128 & 0.121 & 0.136 && 0.055 & 0.052 & 0.050 & 0.053  \\ \hline   
                    
                    &$\beta_{0}$        &  0.106 & 0.104 & 0.104 & 0.112 && 0.050 & 0.049 & 0.049 & 0.051\\ 
{UN}                &$\beta_{1}$        &  0.182 & 0.183 & 0.182 & 0.196 && 0.087 & 0.087 & 0.087 & 0.091\\      
                    &$\beta_{2}$        &  0.115 & 0.111 & 0.110 & 0.118 && 0.054 & 0.050 & 0.052 & 0.052\\  \hline 
\end{tabular}}
\end{table}

\begin{table}[ht]
\centering
\caption{Mean squared error values of parameters for the case of unbalanced longitudinal binary data with different sample sizes when true within-subject correlation level is $\alpha=0.4$.}
\smallskip
\label{tab1S4}
\scalebox{0.7}{
\begin{tabular}{@{}llcccccccccc@{}} \hline
True&&\multicolumn{4}{c}{$N=50$} && \multicolumn{4}{c}{$N=100$}\\
\cline{3-6}\cline{8-11}
correlation &Param.&\multicolumn{4}{c}{Working corr. structure}&& \multicolumn{4}{c}{Working corr. structure}\\
\cline{3-6}\cline{8-11}
structure &&{Indep}     &{AR(1)} & {Exch} & {UN} && {Indep}     &{AR(1)} & {Exch} & {UN} \\
\hline
                    &$\beta_{0}$        & 0.135 & 0.126 & 0.131 & 0.136 && 0.060 & 0.052 & 0.054 & 0.055 \\ 
{AR(1)}             &$\beta_{1}$        & 0.231 & 0.219 & 0.227 & 0.237 && 0.110 & 0.104 & 0.106 & 0.108  \\      
                    &$\beta_{2}$        & 0.119 & 0.097 & 0.105 & 0.108 && 0.060 & 0.049 & 0.055 & 0.053  \\  \hline  
                    
                    &$\beta_{0}$        & 0.159 & 0.149 & 0.145 & 0.158 && 0.075 & 0.072 & 0.069 & 0.076   \\ 
{Exch}              &$\beta_{1}$        & 0.310 & 0.311 & 0.300 & 0.322 && 0.140 & 0.140 & 0.133 & 0.148   \\      
                    &$\beta_{2}$        & 0.117 & 0.093 & 0.084 & 0.090 && 0.059 & 0.049 & 0.042 & 0.048    \\ \hline   
                    
                    &$\beta_{0}$        & 0.130 & 0.122 & 0.124 & 0.129 && 0.065 & 0.063 & 0.062 & 0.068\\ 
{UN}                &$\beta_{1}$        & 0.252 & 0.245 & 0.249 & 0.255 && 0.110 & 0.110 & 0.109 & 0.115 \\      
                    &$\beta_{2}$        & 0.115 & 0.094 & 0.095 & 0.103 && 0.057 & 0.044 & 0.046 & 0.045 \\  \hline 
\end{tabular}}
\end{table}

\begin{table}[ht]
\centering
\caption{Mean squared error values of parameters for the case of balanced longitudinal count data with different sample sizes when true within-subject correlation level is $\alpha=0.2$.}
\smallskip
\label{tab1S5}
\scalebox{0.7}{
\begin{tabular}{@{}llcccccccccc@{}} \hline
True&&\multicolumn{4}{c}{$N=50$} && \multicolumn{4}{c}{$N=100$}\\
\cline{3-6}\cline{8-11}
correlation &Param.&\multicolumn{4}{c}{Working corr. structure}&& \multicolumn{4}{c}{Working corr. structure}\\
\cline{3-6}\cline{8-11}
structure &&{Indep}     &{AR(1)} & {Exch} & {UN}  && {Indep}     &{AR(1)} & {Exch} & {UN} \\
\hline
                    &$\beta_{0}$        & 0.005 & 0.009 & 0.008 & 0.005 && 0.000 & 0.002 & 0.002 & 0.002  \\ 
{AR(1)}             &$\beta_{1}$        & 0.006 & 0.006 & 0.006 & 0.007 && 0.000 & 0.003 & 0.003 & 0.003  \\      
                    &$\beta_{2}$        & 0.005 & 0.005 & 0.005 & 0.005 && 0.000 & 0.002 & 0.002 & 0.002\\   \hline 
                    
                    &$\beta_{0}$        & 0.006 & 0.020 & 0.007 & 0.007 && 0.003 & 0.003 & 0.004 & 0.003    \\ 
{Exch}              &$\beta_{1}$        & 0.008 & 0.008 & 0.008 & 0.008 && 0.004 & 0.004 & 0.004 & 0.004 \\      
                    &$\beta_{2}$        & 0.005 & 0.007 & 0.005 & 0.005 && 0.002 & 0.002 & 0.002 & 0.002  \\  \hline  
                    
                    &$\beta_{0}$        & 0.006 & 0.019 & 0.007 & 0.006 && 0.003 & 0.003 & 0.003 & 0.003 \\ 
{UN}                &$\beta_{1}$        & 0.006 & 0.007 & 0.006 & 0.007 && 0.003 & 0.003 & 0.003 & 0.003  \\      
                    &$\beta_{2}$        & 0.005 & 0.006 & 0.005 & 0.005 && 0.002 & 0.002 & 0.002 & 0.002  \\ \hline  
\end{tabular}}
\end{table}

\begin{table}[ht]
\centering
\caption{Mean squared error values of parameters for the case of balanced longitudinal count data with different sample sizes when true within-subject correlation level is $\alpha=0.4$.}
\smallskip
\label{tab1S6}
\scalebox{0.7}{
\begin{tabular}{@{}llcccccccccc@{}} \hline
True&&\multicolumn{4}{c}{$N=50$} && \multicolumn{4}{c}{$N=100$}\\
\cline{3-6}\cline{8-11}
correlation &Param.&\multicolumn{4}{c}{Working corr. structure}&& \multicolumn{4}{c}{Working corr. structure}\\
\cline{3-6}\cline{8-11}
structure &&{Indep}     &{AR(1)} & {Exch} & {UN}  && {Indep}     &{AR(1)} & {Exch} & {UN} \\
\hline
                    &$\beta_{0}$        &   0.006 & 0.026 & 0.012 & 0.006 && 0.000 & 0.006 & 0.005 & 0.003  \\ 
{AR(1)}             &$\beta_{1}$        &   0.008 & 0.009 & 0.008 & 0.009 && 0.000 & 0.004 & 0.004 & 0.004  \\      
                    &$\beta_{2}$        &   0.004 & 0.007 & 0.005 & 0.004 && 0.000 & 0.002 & 0.002 & 0.002\\   \hline 
                    
                    &$\beta_{0}$        &   0.008 & 0.030 & 0.010 & 0.009 && 0.004 & 0.014 & 0.005 & 0.004 \\ 
{Exch}              &$\beta_{1}$        &   0.012 & 0.012 & 0.011 & 0.013 && 0.006 & 0.006 & 0.006 & 0.006 \\      
                    &$\beta_{2}$        &   0.004 & 0.007 & 0.004 & 0.004 && 0.002 & 0.003 & 0.002 & 0.002 \\ \hline   
                    
                    &$\beta_{0}$        &  0.008 & 0.029 & 0.009 & 0.008 && 0.004 & 0.005 & 0.004 & 0.004 \\ 
{UN}                &$\beta_{1}$        &  0.011 & 0.011 & 0.011 & 0.011 && 0.005 & 0.005 & 0.005 & 0.005 \\      
                    &$\beta_{2}$        &  0.004 & 0.007 & 0.004 & 0.004 && 0.002 & 0.002 & 0.002 & 0.002 \\  \hline 
\end{tabular}}
\end{table}

\begin{table}[ht]
\centering
\caption{Mean squared error values of parameters for the case of unbalanced longitudinal count data with different sample sizes when true within-subject correlation level is $\alpha=0.2$.}
\smallskip
\label{tab1S7}
\scalebox{0.7}{
\begin{tabular}{@{}llcccccccccc@{}} \hline
True&&\multicolumn{4}{c}{$N=50$} && \multicolumn{4}{c}{$N=100$}\\
\cline{3-6}\cline{8-11}
correlation &Param.&\multicolumn{4}{c}{Working corr. structure}&& \multicolumn{4}{c}{Working corr. structure}\\
\cline{3-6}\cline{8-11}
structure &&{Indep}     &{AR(1)} & {Exch} & {UN}  && {Indep} &{AR(1)} & {Exch} & {UN} \\
\hline
                    &$\beta_{0}$        &  0.006 & 0.012 & 0.005 & 0.006 && 0.000 & 0.005 & 0.003 & 0.003  \\ 
{AR(1)}             &$\beta_{1}$        &  0.006 & 0.007 & 0.006 & 0.008 && 0.000 & 0.003 & 0.003 & 0.003  \\      
                    &$\beta_{2}$        &  0.005 & 0.006 & 0.005 & 0.006 && 0.000 & 0.002 & 0.002 & 0.002  \\ \hline   
                      
                    &$\beta_{0}$        &  0.006 & 0.020 & 0.009 & 0.007 && 0.003 & 0.010 & 0.003 & 0.003   \\ 
{Exch}              &$\beta_{1}$        &  0.009 & 0.009 & 0.009 & 0.009 && 0.004 & 0.004 & 0.004 & 0.004   \\      
                    &$\beta_{2}$        &  0.004 & 0.006 & 0.004 & 0.005 && 0.003 & 0.004 & 0.002 & 0.003   \\  \hline  
                    
                    &$\beta_{0}$        &  0.006 & 0.026 & 0.007 & 0.006 && 0.003 & 0.004 & 0.003 & 0.003\\ 
{UN}                &$\beta_{1}$        &  0.007 & 0.008 & 0.007 & 0.008 && 0.004 & 0.004 & 0.004 & 0.004 \\      
                    &$\beta_{2}$        &  0.005 & 0.007 & 0.005 & 0.005 && 0.002 & 0.002 & 0.002 & 0.003\\  \hline 
\end{tabular}}
\end{table}

\begin{table}[ht]
\centering
\caption{Mean squared error values of parameters for the case of unbalanced longitudinal count data with different sample sizes when true within-subject correlation level is $\alpha=0.4$.}
\smallskip
\label{tab1S8}
\scalebox{0.7}{
\begin{tabular}{@{}llcccccccccc@{}} \hline
True&&\multicolumn{4}{c}{$N=50$} && \multicolumn{4}{c}{$N=100$}\\
\cline{3-6}\cline{8-11}
correlation &Param.&\multicolumn{4}{c}{Working corr. structure}&& \multicolumn{4}{c}{Working corr. structure}\\
\cline{3-6}\cline{8-11}
structure &&{Indep}     &{AR(1)} & {Exch} & {UN}  && {Indep}     &{AR(1)} & {Exch} & {UN}\\
\hline
                    &$\beta_{0}$        & 0.007 & 0.029 & 0.010 & 0.008 && 0.000 & 0.012 & 0.004 & 0.003   \\ 
{AR(1)}             &$\beta_{1}$        & 0.008 & 0.009 & 0.008 & 0.009 && 0.000 & 0.004 & 0.004 & 0.004    \\      
                    &$\beta_{2}$        & 0.005 & 0.007 & 0.005 & 0.005 && 0.000 & 0.003 & 0.002 & 0.002   \\   \hline 
                    
                    &$\beta_{0}$        & 0.009 & 0.029 & 0.013 & 0.009 && 0.005 & 0.015 & 0.009 & 0.004  \\ 
{Exch}              &$\beta_{1}$        & 0.012 & 0.013 & 0.012 & 0.013 && 0.006 & 0.006 & 0.006 & 0.006   \\      
                    &$\beta_{2}$        & 0.005 & 0.007 & 0.004 & 0.004 && 0.002 & 0.004 & 0.003 & 0.002   \\  \hline  
                    
                    &$\beta_{0}$        & 0.008 & 0.034 & 0.010 & 0.008 && 0.004 & 0.007 & 0.004 & 0.004\\ 
{UN}                &$\beta_{1}$        & 0.011 & 0.011 & 0.011 & 0.012 && 0.005 & 0.005 & 0.005 & 0.005 \\      
                    &$\beta_{2}$        & 0.005 & 0.008 & 0.004 & 0.005 && 0.002 & 0.002 & 0.002 & 0.002 \\  \hline 
\end{tabular}}
\end{table}

\end{document}